\documentclass{article}

\usepackage{lmodern}

\usepackage[margin=1.2in]{geometry}
\usepackage[utf8]{inputenc} 
\usepackage[T1]{fontenc}    
\usepackage{hyperref}       
\usepackage{url}            
\usepackage{booktabs}       

\usepackage{amsmath,amsfonts,amsthm,amssymb}

\usepackage{nicefrac}       
\usepackage{microtype}      
\usepackage{mathabx}
\usepackage{mathtools}
\usepackage{xfrac}
\usepackage{csquotes}
\usepackage{algorithm}
\usepackage{algpseudocode}
\usepackage{makecell}
\usepackage{bbm}

\usepackage[]{xcolor}
\usepackage[numbers]{natbib}

\newcommand{\EE}{\mathbb{E}}

\newcommand{\ptrue}{q_{\mathrm{lim}}}
\newcommand{\pchoice}{p}
\newcommand{\pmax}{p_O}
\newcommand{\pt}{p_t^*}
\newcommand{\pagg}{Q^{\mathrm{agg}}}
\newcommand{\pbc}{Q^{\mathrm{BC}}}
\newcommand{\pperm}{Q^{\mathrm{perm}}}
\newcommand{\plim}{Q_{\mathrm{lim}}}
\newcommand{\pval}{Q}
\newcommand\independent{\protect\mathpalette{\protect\independenT}{\perp}}\def\independenT#1#2{\mathrel{\rlap{$#1#2$}\mkern2mu{#1#2}}}

\newcommand{\tabitem}{~~\llap{\textbullet}~~}

\newtheorem{theorem}{Theorem}[section]
\newtheorem{corol}[theorem]{Corollary}
\newtheorem{proposition}[theorem]{Proposition}
\newtheorem{lemma}[theorem]{Lemma}

\newtheorem{remark}{Remark}
\newtheorem{example}{Example}

\usepackage{graphicx} 
\usepackage{caption}
\usepackage{subcaption}
\usepackage{cleveref}

\usepackage{changes}

\title{Sequential Monte-Carlo testing by betting}

\author{%
  Lasse Fischer and Aaditya Ramdas\\
   University of Bremen and Carnegie Mellon University \\
  \texttt{fischer1@uni-bremen.de, aramdas@cmu.edu}
}

\begin{document}

\maketitle
\begin{abstract}
    In a Monte-Carlo test, the  observed dataset is fixed, and several resampled or permuted versions of the dataset are generated in order to test a null hypothesis that the original dataset is exchangeable with the resampled/permuted ones. Sequential Monte-Carlo tests aim to save computational resources by generating these additional datasets sequentially one by one, and potentially stopping early. While earlier tests yield valid inference at a particular prespecified stopping rule, our work develops a new anytime-valid Monte-Carlo test that can be continuously monitored, yielding a p-value or e-value at any stopping time possibly not specified in advance. It generalizes the well-known method by Besag and Clifford, allowing it to stop at any time, but also encompasses new sequential Monte-Carlo tests that tend to stop sooner under the null and alternative without compromising power. The core technical advance is the development of new test martingales (nonnegative martingales with initial value one) for testing exchangeability against a very particular alternative. These test martingales are constructed using new and simple betting strategies that smartly bet on whether a generated test statistic is greater or smaller than the observed one. The betting strategies are guided by the derivation of a simple log-optimal betting strategy, have closed form expressions for the wealth process, provable guarantees on resampling risk, and display excellent power in practice. 
\end{abstract}

\tableofcontents

\section{Introduction}

Suppose we have some data $X_0$, from which we calculate a test statistic $S(X_0)$. Suppose further that we have the ability to generate additional data $X_1,X_2,\dots$ and calculate test statistics $S(X_1),S(X_2),\dots$~.  Denoting $Y_i = S(X_i)$, assume further that $Y_1,Y_2,\dots$ are always exchangeable conditional on $Y_0$. Suppose that we are trying to test the null hypothesis
\[
H_0: Y_0, Y_1, \dots \text{ are exchangeable.}
\]
This type of null appears in permutation tests (for example, when performing two-sample or independence testing), but also in many other settings. While we focus on permutation tests for concreteness, the reader may safely substitute \enquote{Monte-Carlo} for \enquote{permutation} in all occurrences. In particular the methods also apply without alteration to conditional randomization tests  frequently encountered in causal inference, and Monte-Carlo tests of group invariance for any other group actions beyond permutations (such as sign symmetry or rotation tests). The p-value corresponding to $H_0$ would be
\begin{equation}
\label{eq:perm-pval}
\pperm_T := \frac{1+\sum_{t=1}^T \mathbbm{1}\{Y_t \geq Y_0\}}{T+1}
\end{equation}
for some prespecified number of permutations $T$.
But how should one choose $T$? In practice, it is often just set as a suitably large constant like 1000, or sometimes even larger if multiple testing corrections need to be performed. The goal of this paper is to describe a simple ``anytime-valid p-value'' (or ``p-process'') and an ``e-process'' for $H_0$ that is obtained using the principle of testing by betting.

Before we do this, it makes sense to first describe the (working) alternative: 
\[
H_1: Y_1, Y_2,\dots \text{ are exchangeable conditional on } Y_0,\text{ but } Y_0 \text{ is stochastically larger than } Y_1, Y_2,\dots\text{ }.
\]
We used the word ``working'' above to highlight that when the null is false, we will not require or assume that the above $H_1$ is true. Since the complement of $H_0$ is enormous, no test can be powerful against all alternatives, and specifying such a ``working'' alternative helps guide the design of powerful and practical tests.

While there do exist sequential tests for exchangeability, they have focused on (for example) changepoint \citep{vovk2021testing} or Markov alternatives \citep{ramdas2022testing, saha2023testing}. The technique introduced in this paper is different, and focuses on the above $H_1$ that is practically relevant for settings like two-sample or independence testing using permutations.

\begin{remark}
    In order to avoid ties, each test statistic $Y_i$ should be considered as coming together with a random sample $\theta_i$ from $U[0,1]$ which is independent of all test statistics and all other random samples \citep{vovk2003testing,vovk2005algorithmic, vovk2021testing}. If two test statistics are equal $Y_i=Y_j$, then comparing $\theta_i,\theta_j$ decides which test statistic is considered as \enquote{greater}. Hence, the event $\{Y_t\geq Y_0\}$ should be read as $\{Y_t>Y_0\} \cup \{Y_t=Y_0, \theta_t>\theta_0\}$. To the detriment of power, drawing random samples can be avoided when a conservative choice can be made in case of a tie. For instance, if we count all ties between $Y_0$ and a generated test statistics $Y_t$ as $Y_t\geq Y_0$, the permutation p-value $\pperm$ would still be valid, but conservative. The same holds for our binomial mixture strategy with uniform density presented in Section \ref{sec:closed_form_binomial_mixture}.
\end{remark}

\begin{remark}\label{remark:randomization}
    In a permutation test one can either draw the permutations without replacement or with replacement. If it is computationally feasible to draw all possible permutations, one usually does the former and sets the parameter $T$ in~\eqref{eq:perm-pval}
    to equal the total number of possible permutations. However, it is computationally more convenient to draw the permutations with replacement. In the hypothesis formulation above it is implicitly assumed that we have access to an infinite number of generated test statistics and therefore that the permutations are drawn with replacement. This allows us to characterize the asymptotic behavior in the remainder of the paper. However, our proposed methods are also valid if permutations are sampled without replacement and/or a maximum number of permutations is fixed in advance.
\end{remark}

\begin{example}[Testing independence in a treatment vs. control trial using a permutation test]\label{example:Fisher}
    Suppose we observe some data $X_0$ of a (randomized) treatment vs.\ control trial, where $X_0$ consists of the responses $X$ and  binary treatment indicators $W_0$ (both being $n$-dimensional vectors if there are $n$ patients).  We wish to test whether the treatment has no effect on the response, which can also be written as
    $$
    H_0^{\mathrm{IT}}: W_0 \independent X.
    $$
    We can generate new datasets $X_1=(X,W_1),X_2=(X,W_2),\ldots$ by randomly permuting $W_0$. Note that we only permute the treatment labels and not the observed response data $X$. In this way, the datasets $X_0,X_1,\ldots$  are exchangeable under the null hypothesis, but $X_0$ is not exchangeable with $X_1,X_2,\ldots$ if the null hypothesis is not true. In addition, the generated datasets $X_1,X_2,\ldots$ are always exchangeable (even conditional on $X_0$), since they were generated in an identical manner. The same properties hold for $Y_0,Y_1,\ldots$, where $Y_i=S(X_i)$ and $S$ can be any test statistic --- a common choice is the difference in mean between the treated and untreated observations. Hence, this example fits perfectly into our introduced setting, and the reader may keep it in mind in what follows. We also consider it in our simulations (Section \ref{sec:sim_p-value}--\ref{sec:sim_stochastic_rounding}) and a real data application (Section \ref{sec:real_data}). 
\end{example}

\begin{example}[Testing conditional independence with the CRT]\label{example:CRT}
    Consider testing the independence between two variables $W_0$ and $X$ conditional on some covariates $Z$, denoted by $$H_0^{\mathrm{CI}}: W_0 \independent X \ |\ Z.$$
    For example, $W_0$ could be a binary treatment vector, $X$ the observed responses and $Z$ contains the patient covariates.
If the distribution of $W_0|Z$ is known, one can test $H_0^{\mathrm{CI}}$ using a conditional randomization test (CRT); while this is a common idea in causal inference, \citet{candes2018panning} discuss broader applications. For example, $W_0 | Z$ is indeed known if we are analyzing data from a designed experiment where the treatment was randomly assigned to each patient based on a known probability that depends on that patient's covariates. Here, we cannot permute $W_0$ because the subjects are not exchangeable conditional on $Z$, but we can redraw $W_0$ as follows.
Let $W_1,W_2,\ldots$ be sampled independently from the distribution of $W_0|Z$  and $Y_t:=S(W_t,X,Z)$, $t\in \mathbb{N}_0$, for an arbitrary test statistic $S$. The test statistics $Y_1,Y_2,\ldots$ are always exchangeable conditional on $Y_0$, while $Y_0,Y_1,\ldots$ are exchangeable under $H_0^{\mathrm{CI}}$, matching the setting of our paper. The CRT p-value essentially uses the same formula as the permutation p-value~\eqref{eq:perm-pval} (even though it is technically not a permutation test, but a randomization test). Generating samples from $W_0|Z$ can be computationally challenging or recalculating the test statistic on these samples can be very costly, so a lot of computational effort can potentially be saved by using our sequential permutation tests. We revisit this example in Section \ref{sec:real_CRT}.
\end{example}

\subsection{Our approach at a high level\label{sec:high_level}}

We design a game of chance, where a gambler begins with one dollar, makes bets in each round so that their wealth changes over time. The value of the wealth at any point will be a measure of evidence against the null. Throughout the paper, 
$\mathbb{P}$ denotes the true but unknown probability distribution and $\mathbb{E}$  the corresponding expected value, whereas $\mathbb{P}_{H_0}$ and $\mathbb{E}_{H_0}$ denote the probability and expectation under the null hypothesis (i.e. corresponding to an arbitrary distribution that satisfies $H_0$). 

The game will proceed in rounds, which we index by $t=1,2,\dots$ --- in round $t$, the gambler will first bet on the outcome of $I_t:=\mathbbm{1}\{Y_t\geq Y_0\}$. We will soon specify what properties this bet must satisfy. Then nature will draw $X_t$ and calculate $Y_t$ and reveal $I_t$ (hiding $Y_t$ itself), thus determining our payoff for that round. In the following we will also refer to the event $\{I_t=1\}$ as a \enquote{loss} and to $\{I_t=0\}$ as a \enquote{win}. This is because under the alternative, it is intuitive that the gambler will usually put more money on the event $\{I_t=0\}$, but we also theoretically justify this in Section~\ref{sec:log-optimal}. Similarly, the random variable $L_T:=\sum_{t=1}^T I_t$ will be called the number of losses after $T$ permutations. 

It is important that we do not know the raw $Y$ values, but only the indicators $I$ are revealed in each round. So after the first round, we only know if $Y_1$ was larger or smaller than $Y_0$, but not the values of either $Y_0$ or $Y_1$.

The gambler's bet will be specified by a betting function $B_t:\{0,1\}\rightarrow \mathbb{R}_{\geq 0}$. $B_t(r)$ denotes the amount their wealth gets multiplied by when $I_t=r$, $r\in \{0,1\}$. Their wealth after $T$ rounds of betting is
\[
W_T := \prod_{t=1}^T B_t(I_t).
\]

Denoting $I_1^t=\{I_1,\dots,I_t\}$, our bet will have to satisfy the constraint $\EE_{H_0}[B_t(I_t) \mid I_1^{t-1}]=1$ so that the wealth process is a test martingale (a nonnegative martingale with initial value one under $H_0$). Of course, this implies that only one of $B_t(0)$ and $B_t(1)$ can be larger than 1, and typically this will be $B_t(0)$ (as guided by $H_1$). To simplify the above constraint, note that conditional on $I_1^{t-1}$, the event $\{I_t=1\}$ is equivalent to the event that $Y_t$ ranks among the $L_{t-1}+1$ largest test statistics of $Y_0,\ldots,Y_t$. Since $Y_0,\ldots, Y_t$ are exchangeable under $H_0$, the constraint simplifies to 
\begin{align}
B_t(0)\frac{t-L_{t-1}}{t+1}+B_t(1)\frac{1+L_{t-1}}{t+1}=1. \label{eq:cond_bet}
\end{align}

As a consequence of the optional stopping theorem for nonnegative martingales, the wealth being a test martingale implies that it is an ``e-process''~\citep{ramdas2022testing,ramdas2023game}, which is a nonnegative sequence of random variables that satisfies
\begin{equation}\label{eq:e-process}
\mathbb E_{H_0}[W_\tau] \leq 1,
\end{equation}
for any stopping time $\tau$ with respect to the filtration $(\mathcal I_t)_{t \geq 1}$ of the indicators $\mathcal I_t := \sigma(I_1,\dots,I_t)$. Recall that an integer-valued random variable $\tau$ is called stopping time, if $\{\tau=t\}$ is measurable with respect to $\mathcal{I}_t$ for all $t\in \mathbb{N}$.

Said differently,~\eqref{eq:e-process} states that
$W_\tau$ is an e-value at any stopping time $\tau$ that is possibly not even specified in advance. E-processes capture evidence against the null: the larger the wealth, the more the evidence against the null. If the null is true,~\eqref{eq:e-process} implies that no betting strategy and stopping rule is guaranteed to make money (i.e., demonstrate evidence against the null). However, if the alternative is true, then a good betting strategy can make money and the evidence against the null can grow over time (as more permutations are drawn). We will design such good betting strategies in this paper.  

Ville's martingale inequality states that for any $\alpha\in(0,1)$, 
\begin{equation}\label{eq:ville}
\mathbb P_{H_0} (\exists t \geq 1: W_t \geq 1/\alpha) \leq \alpha,
\end{equation}
and thus if we wish to reject the null at level $\alpha$, we may stop as soon as the wealth exceeds $1/\alpha$. For example, if $\alpha=0.05$, we reject the null if we were able to turn our initial (toy) dollar into twenty dollars. However, this is only one possible stopping rule, and our e-process can be interpreted as the evidence available against the null at \emph{any} stopping time. For example, suppose we stop (for whatever reason) when our wealth is 16, that is still reasonable evidence against the null (indeed we would be quite impressed with a gambler who went to a casino and multiplied their wealth 16-fold), and one does not require a predefined level or threshold to be able to interpret it. 

Finally, since readers may be more familiar with p-values,~\eqref{eq:ville} implies that the process defined by 
\begin{align}
\pval_t := 1/(\sup_{s \leq t} W_s) \label{eq:p-process}
\end{align}
is an anytime-valid p-value, or ``p-process''~\citep{johari2022always,howard2021time,ramdas2023game}, meaning that at any arbitrary stopping time $\tau$, $\pval_\tau$ is a p-value:
\[
\mathbb P_{H_0}(\pval_\tau \leq \alpha) \leq \alpha.
\]
We notate the p-values with the capital letter $\pval$ and their realizations with the lowercase letter $q$ to distinguish them from the parameters $p$ introduced later.

\begin{algorithm}
\caption{General strategy for permutation testing by betting} \label{alg:general}
 \hspace*{\algorithmicindent} \textbf{Input:} Sequence of test statistics $Y_0,Y_1, Y_2, \ldots$.\\
 \textbf{Optional Input:} Stopping rule $\mathcal S$ (potentially data-dependent and decided on the fly).\\
 \hspace*{\algorithmicindent} \textbf{Output:} E-process $(W_t)_{t \geq 1}$, and p-process $(1/\sup_{s \leq t} W_s)_{t \geq 1}$.\\
 \textbf{Optional output:} Stopping time $\tau$, e-value $W_{\tau}$, p-value $1/\sup_{s \leq \tau} W_s$.
\begin{algorithmic}[1]
\State $W_0 = 1$ 
\State $L_0=0$
\For{$t=1,2,...$}
\State Choose a bet $B_t=(B_t(0),B_t(1))$ with $B_t(0)\frac{t-L_{t-1}}{t+1}+B_t(1)\frac{1+L_{t-1}}{t+1}=1$
\State Reveal $I_t$
\State Calculate $L_t=L_{t-1} + I_t$
\State $W_t = W_{t-1} \cdot B_t(I_t)$ 
\If{$\mathcal S(I_1,\ldots,I_t)=\text{stop}$} 
    \State $\tau=t$
    \State \Return $\tau, W_{\tau}, (\max_{s=1,\ldots,\tau} W_s)^{-1}$
\EndIf
\EndFor
\end{algorithmic}
\end{algorithm}

\begin{remark}\label{remark:vovk}
    A general approach for the sequential testing of exchangeability has been proposed by Vovk and colleagues \citep{vovk2005algorithmic, vovk2003testing, vovk2021testing} based on conformal prediction. Instead of betting at each step $t$ on the indicator $I_t=\mathbbm{1}\{Y_t\geq Y_0\}$, Vovk considers the rank $R_t=\sum_{i=0}^{t} \mathbbm{1}\{Y_i\geq Y_t\}$ of $Y_t$ among $Y_0,\ldots, Y_t$. We could have proceeded in the same way. However, since $Y_1,Y_2,\ldots$ are exchangeable conditional on $Y_0$ under our alternative, the rank of $Y_t$ among $Y_1,\ldots, Y_t$ is uniform (even conditional on $R_1,\ldots, R_{t-1}$). Therefore, the only relevant information is whether $Y_t\geq Y_0$. Another subtle difference to Vovk's approach is the randomization technique used to handle ties (see Remark~\ref{remark:randomization}). In Vovk's approach, one has to draw a random sample $\theta_t$ for every test statistic $Y_t$, while we only require a random sample $\theta_t$ if $Y_t$ and $Y_0$ actually tie. This reduces the computational effort, the main goal of this paper. In addition, with Vovk's randomization approach one must specify a bet for each value in $[0,1]$, while we only bet on $\{0,1\}$ which simplifies the betting function $B_t$ and leads to closed forms for the wealth $W_T$, as we will illustrate in the remainder of the paper.
    
    To summarize, our approach can be considered as a special case of Vovk's conformal martingale approach. However, our amendments simplify the method substantially for our particular alternative --- an alternative, which to the best of our knowledge has not been stated in the literature before and naturally arises from the classical problem of sequential permutation testing. The development of powerful tests for this alternative requires the construction of new betting strategies.
\end{remark}

\subsection{Related work}
The use of test martingales and  \enquote{testing by betting} have been popularized as a safe approach to sequential, anytime-valid inference \citep{shafer2011test, shafer2021testing,howard2021time,grunwald2020safe,waudby2023estimating}. 
An overview of this new methodology and the recent advances is given by \citet{ramdas2023game}.

There have been several other works on sequential tests of exchangeability~\citep{vovk2021testing,ramdas2022testing,saha2023testing,koning2023online, lardy2024anytime}. The most important differentiating factor of the current paper from those is the alternative hypothesis being tested. The aforementioned works involve the \emph{data} coming in sequentially, and the alternative involves either (a) a changepoint from exchangeability (i.e., the first $j$ samples are exchangeable, but not the following samples, with $j=0$ being a special case), (b) Markov alternatives (i.e., there is some serial time-ordered dependence between the data).

In contrast, our paper considers the setting where the data are fixed, the permutations are drawn sequentially, and thus --- by design --- the alternative is such that the first test statistic is special while the rest are exchangeable. 

In some sense, our paper is really targeted at a sequential version of the classical problem of permutation testing, while the earlier papers are essentially trying to test if the data are i.i.d., or if there is some distribution shift or drift. Despite their different goals, we describe some key technical aspects of these papers in Appendix~\ref{appsec:related-work}. A more detailed description of the relation to the conformal prediction approach by Vovk and colleagues \citep{vovk2003testing, vovk2005algorithmic, vovk2021testing} is given in Remark~\ref{remark:vovk}. Another related martingale-based approach by \citet{waudby2020confidence} constructs confidence sets for some parameter by sampling without replacement (WoR). However, this problem implicitly tests a different null hypothesis, which leads to a dissimilar method. 

The closest related work to ours is the famous sequential Monte-Carlo p-values method of \citet{besag1991sequential}, who introduced a simple strategy for reducing the number of permutations drawn that is still very popular today. It can be shown that for a certain parameter choice, their method leads to the same decisions as the classical permutation test while needing less permutations \citep{silva2009power}. Inspired by this, further sequential Monte-Carlo tests with the objective to reduce the number of permutations were proposed \citep{fay2007using, silva2013optimal}. An important difference to our approach is that theirs are stopped based on fixed predefined rules, and thus yield valid p-values only at these particular stopping times, and are not valid at any other stopping time. In addition, the existing approaches do not have the option to continue testing by drawing more permutations once the stopping criterium is reached (for example, if the null is not rejected), while our method can always handle drawing additional permutations and stopping at a later point, even if these flexibilities are not specified in advance.

Our general betting approach generalizes the Besag-Clifford method, but also allows to construct new strategies. Our simulations suggest that these new strategies can stop sooner under both the null and the alternative, without sacrificing power. 

There are also other approaches to reduce the computational cost of permutation tests than sequential permutation p-values. \citet{koning2024more} and  \citet{koningHemerik2024more} consider gaining computational efficiency by considering a subgroup of permutations. Another line of work tries to approximate the null distribution of $Y_0$ analytically \citep{robinson1982saddlepoint, davison1988saddlepoint, niu2024computationally}. \citet{gandy2009sequential} proposes to bound the probability of obtaining a different decision than the limiting permutation p-value. A recent and comprehensive overview of the literature on computational efficient permutation tests --- including an assessment of the approach introduced in this paper --- is given by \citet{stoepker2024inference}.

\subsection{Outline of the paper and our contributions}

In this paper we construct anytime-valid permutation tests based on  novel test martingales. Our proposed e-values and p-values have easy calculable closed forms, save computational efforts under both the null and the alternative and can bound the resampling risk at arbitrary small values.
We start with two simple betting strategies which serve to exemplify the approach and provide some first intuition of  the opportunities it offers (Section~\ref{sec:aggressive_strat}). Afterwards, we derive the (oracle) log-optimal betting strategy, which depends on an unknown parameter (Section~\ref{sec:log-optimal}). Based on these results, we introduce in Section~\ref{sec:our_strat} a simple strategy that is based on a single parameter indicating aggressiveness and which leads to a wealth that is proportional to the likelihood of a binomial distribution. In Section~\ref{sec:average_wealth}, we show how the log-optimal strategy can be mimicked by a mixture of these simple binomial strategies and use this to introduce a certain working prior under which the strategy always rejects when the limiting permutation p-value $\plim$ is less than $c$ for some constant $c<\alpha$ arbitrary close to $\alpha$. In Section \ref{sec:betting_is_general}, we show that every permutation test based on the number of losses can be obtained by our general betting approach. In particular, we derive concrete betting strategies that yield anytime-valid generalizations of the classical permutation p-value and the Besag-Clifford method. We compare our approach to the Besag-Clifford strategy \citep{besag1991sequential} and the classical permutation p-value via simulations and real data analyses in Section~\ref{sec:simulations}. The R code to reproduce all figures and results is available at the GitHub repository \url{github.com/fischer23/MC-testing-by-betting}.

\section{The aggressive strategy}
\label{sec:aggressive_strat}



We begin by discussing two simple strategies to provide some intuition about our approach. In Sections \ref{sec:log-optimal}--\ref{sec:average_wealth} we then introduce our recommended strategies based on maximizing the expected logarithm of the wealth.
The most passive betting strategy uses the betting function
\begin{equation}\label{eq:pass}
B_t(0)=B_t(1)=1 \quad (t\in \mathbb{N})
\end{equation}
 and thus no matter the outcome, our wealth is multiplied by 1. Thus, the wealth starts at one and stays at 1, meaning that the resulting test is powerless.

In contrast, the most aggressive betting strategy uses the function
\begin{equation}\label{eq:agg}
B_t(0)=\frac{t+1}{t} \quad \text{and} \quad B_t(1)=0 \quad (t\in \mathbb{N}),
\end{equation}
as long as $L_{t-1}=0$.
 This would be the strategy of a gambler who not only strongly believes that $H_1$ is true, but also  believes that $Y_0$ is the largest amongst all other $Y_i$'s, and so the gambler deems it is (nearly) impossible for $Y_{t}$ to be smaller than $Y_0$, forsaking all his money if that happens. This is a risky strategy: the first round $t$ in which $Y_{t}$ ends up larger than $Y_0$, we would multiply our wealth by 0 and can never recover from it. However, if our gamble was actually always correct, then our wealth after $T$ rounds would  be $\prod_{t=1}^T (t+1)/t = T+1$, and this is the largest possible wealth. Indeed, since the inverse of the wealth is an anytime-valid p-value, it would yield a p-value of $1/(T+1)$, the best possible in this case.



\subsection{Connection to Besag-Clifford\label{sec:besag_clifford}}

The aggressive strategy yields an anytime-valid generalization of the  sequential permutation test by~\citet{besag1991sequential}, when their method is stopped after one loss. Define the stopping time $\gamma(h,T)$ that stops after $h$ losses or after sampling $T$ (possibly infinite) permutations, whichever happens first. Recalling our notation $L_t$ for the number of losses after $t$ steps, Besag and Clifford defined a single p-value $\pbc_{\gamma(h,T)}$ (as opposed to a p-process) as 
\begin{align}\pbc_{\gamma(h,T)}=\begin{cases}
            h/\gamma(h,T), & L_{\gamma(h,T)}=h \\
            (L_{\gamma(h,T)}+1)/(T+1), & \text{otherwise}.
        \end{cases} \label{eq:bc_pval}\end{align}
The p-value $\pbc_{\gamma(h,T)}$ is valid only at time $\gamma(h,T)$ and not at any other stopping time. While Besag and Clifford mentioned the possibility of $T=\infty$, they mainly focused on $T<\infty$ and most follow-up work ignored $T=\infty$ entirely. For this reason, we briefly discuss the case $T=\infty$ in Appendix \ref{sec:neg_bin}.

Since our test yields a p-value at any stopping time, it makes sense to ask what the p-process \eqref{eq:p-process} defined by the aggressive strategy $(\pagg_t)_{t\in \mathbb{N}}$ equals at time $\gamma(1,T)$. The fascinating answer is that it equals the Besag-Clifford p-value at that stopping time:
\[
\pagg_{\gamma(1,T)} = \pbc_{\gamma(1,T)}.
\]
The above equality has two key implications. First, it clearly leads to the interpretation that our aggressive betting strategy yields a p-process that is an anytime-valid generalization of the Besag-Clifford p-value that stops after one loss. (Our p-process recovers the latter p-value exactly at time $\gamma(1,T)$, but remains valid at all other stopping times.) Second, it is well known that $\pbc_{\gamma(1,T)}$ is in fact an \emph{exact} p-value, meaning that even though it can only take on the values $1/1, 1/2, 1/3,\dots, 1/(T+1)$, we have $\mathbb P(\pbc_{\gamma(1,T)} \leq c) = c$ for any $c$ in this support. Thus, the same property also holds for $\pagg_{\gamma(1,T)}$. 

In Section~\ref{sec:betting_is_general}, we provide an anytime-valid generalization of the Besag-Clifford p-value for general $h$, using a more sophisticated betting strategy.

\subsection{Resampling risk\label{sec:resampling_risk}}

 Since $Y_1,Y_2,\ldots$ are always exchangeable conditional on $Y_0$, it follows that $I_1, I_2,\ldots$ are also exchangeable. De Finetti's representation theorem implies that $L_t/t \stackrel{a.s.}{\to} \plim$ for $t \to \infty$, where $\plim$ is some random variable with values in $[0,1]$ and $I_1,I_2,\ldots$ are i.i.d. conditional on $\plim$ with $\mathbb{P}(I_1=1|\plim)=\plim$. 
 If $H_0$ is true, then $\plim$ follows a uniform distribution on $[0,1]$. If $H_1$ is true, then $\plim$ is stochastically smaller than uniform. Consequently, the hypotheses $H_0$ and $H_1$ imply the following null $H_0'$ and alternative $H_1'$, respectively, formulated in terms of the indicators $I_t$ instead of the test statistics $Y_t$
 \begin{align*}
  H_0':I_1,I_2,\ldots\text{ are i.i.d. } \mathrm{Ber}(\ptrue) \text{ cond. on } \plim=\ptrue \text{ with } \plim \sim U[0,1],  
 \end{align*}
 $$
 H_1':I_1,I_2,\ldots\text{ are i.i.d. } \mathrm{Ber}(\ptrue) \text{ cond. on } \plim=\ptrue \text{ with } \plim \text{ stochastically smaller than } U[0,1].
 $$
 
 The \emph{limiting permutation p-value} $\plim$ can also be interpreted as the p-value we would choose if we knew the distribution of $Y_0$ under $H_0$. Hence, a reasonable
  objective in the Monte Carlo test literature is to find a test  $\phi\in \{0,1\}$ that bounds or minimizes the probability of obtaining a different decision than $\plim$. This is called resampling risk \citep{fay2002designing}, and is defined by 
    $$\mathrm{RR}_{\ptrue}(\phi)=\begin{cases}
            \mathbb{P}(\phi=0|\plim=\ptrue), & \ptrue\leq \alpha \\
            \mathbb{P}(\phi=1|\plim=\ptrue), & \ptrue > \alpha.
        \end{cases}$$
     Note that only the randomness induced by the resampling mechanism is relevant for the resampling risk $\mathrm{RR}_{\ptrue}$, since $\plim=\ptrue$ is fixed. A small resampling risk guarantees that running the same permutation test (with different random seed) on the same data will likely yield the same result. \citet{gandy2009sequential} introduced an algorithm that uniformly bounds $\mathrm{RR}_{\ptrue}$ for all $\ptrue \in[0,1]$ by arbitrary small $\epsilon >0$ and which stops after a finite number of steps if $\ptrue\neq \alpha$. In this paper, we are mainly concerned with the case of $\ptrue \leq \alpha$, since our tests are valid by construction (which Gandy's are not) such that additional rejections can be seen as power improvement and do not pose a problem for validity. 
     
     The aggressive strategy leads to a resampling risk of $\mathrm{RR}_{\ptrue}(\mathbbm{1}\{\pagg_{\tau}\leq \alpha\})=1-(1-\ptrue)^{\lceil\frac{1}{\alpha}-1\rceil}\approx (\ptrue-\alpha \ptrue)/\alpha$ for $\ptrue \leq \alpha $, which is greater than zero for all $\ptrue\neq 0$. Note that $\ptrue$ never equals exactly $0$ (since there is a positive probability of drawing the identity) and therefore the resampling risk of the aggressive strategy is never zero. In Section~\ref{sec:average_wealth}, we introduce a ``binomial mixture strategy'' with zero resampling risk for all $\ptrue \in [0,c)$, where $c<\alpha$ can be chosen arbitrarily close to $\alpha$. Along the way, we introduce a ``binomial strategy'' with a zero resampling risk for all $\ptrue \in [0,\alpha/C)$ for an (explicit) constant $C$.
     
     Note that the resampling risk of the permutation p-value and Besag-Clifford method equals zero only if $\ptrue\in \{0,1\}$ --- and therefore is practically never zero. To the best of our knowledge, our methods are the first yielding zero resampling risk for a nontrivial set of $\ptrue$ values in the literature (whether using anytime-valid tests or not).

\section{The oracle log-optimal strategy}
\label{sec:log-optimal}

Suppose we knew the true distribution of $I_t|I_1^{t-1}$, $t\in \mathbb{N}$, how should we place our bets $B_t$ to maximize the wealth $W_T$? This question is not easy to answer, since $W_T$ is a random variable and thus it is not clear in which sense it should be maximized. Intuitively, one could think of maximizing $\mathbb{E}[W_T]$. However, in this case we may end up with a strategy that leads to low or even zero wealth in many cases and to a very large wealth in a few cases. This seems not desirable. In order to avoid this, in similar problems, such as in portfolio theory or betting on horse races \citep{cover1999elements}, the objective is to maximize the expected logarithmic wealth: \begin{align}\mathbb{E}[\log(W_T)]=\sum_{t=1}^T \mathbb{E}[\log(B_t(I_t))].\label{eq:log_obj}\end{align} 
This is also called the Kelly criterion \citep{kelly1956new, breiman1961optimal}. The expected log-wealth has now become the de facto standard measure of performance in testing by betting; see~\cite{shafer2021testing,grunwald2020safe,waudby2023estimating,ramdas2023game}. Despite not being the only option, we see no particular reason to deviate from it, so we adopt this perspective going forward in guiding our design of good betting strategies. However, it should be noted that log-optimality does not imply optimality of the corresponding p-process \eqref{eq:p-process}.

The log-wealth can be maximized by maximizing $\mathbb{E}[\log(B_t(I_t))|I_1^{t-1}]$ at each step $t\in \mathbb{N}$. Note that we need to condition on the past, since we are allowed to choose our betting strategy based on the previous loss indicators. A betting strategy that maximizes \eqref{eq:log_obj} is called a log-optimal strategy. One main reason to consider the logarithmic wealth is that the sum offers nice asymptotic behavior. For example, if the considered random variables $Z_1,\ldots, Z_T$ are i.i.d., the strong law of large numbers implies that $$\frac{1}{T}\sum_{t=1}^T \log(Z_1) \to \mathbb{E}[\log(Z_1)] \text{ with probability } 1.$$
 But for short time periods and when the i.i.d. assumption is violated, log-optimal strategies still offer desirable properties \citep{cover1999elements}. Therefore, maximizing the logarithmic wealth seems like a reasonable objective, although we consider a non-i.i.d. betting process. 


The loss indicators $I_1, I_2,\ldots$ are exchangeable under the null and the alternative and therefore the vector $I_1^{t-1}$ does not contain more information about $I_t$ than $L_{t-1}=\sum_{i=1}^{t-1} I_i$. Therefore, we replace $I_1^{t-1}$ by $L_{t-1}$ in the condition $\mathbb{P}(I_t=1|I_1^{t-1})$ from now on. In the following proposition we show a basic result that characterizes the log-optimal strategy.

\begin{proposition}
    Let $\pt=\mathbb{P}(I_t=1|L_{t-1})$, $t\in \mathbb{N}$. Then the (unique) log-optimal strategy is given by 
    \begin{align}
        B_t^*(0)=(1-\pt)\frac{t+1}{t-L_{t-1}} \quad \text{and} \quad B_t^*(1)=\pt\frac{t+1}{L_{t-1}+1.} 
        \label{eq:log_opt_strategy}
    \end{align}
 \label{theo:log_optimal}
\end{proposition}

\begin{proof}
    Let $b_t= B_t(1)(1+L_{t-1})/(t+1)$ for some betting strategy $B_t$. We obtain by Gibbs' inequality
    \begin{align*}
        &\mathbb{E}[\log (B_t(I_t))|L_{t-1}] \\ &=  \pt \log (B_t(1)) + (1-\pt) \log (B_t(0)) \\
        &= \pt\log\left(\frac{t+1}{1+L_{t-1}}\right)+ (1-\pt)\log\left(\frac{t+1}{t-L_{t-1}}\right)+\pt\log(b_t)+(1-\pt)\log(1-b_t) \\
        &\leq \pt\log\left(\frac{t+1}{1+L_{t-1}}\right)+ (1-\pt)\log\left(\frac{t+1}{t-L_{t-1}}\right)+\pt\log(\pt)+(1-\pt)\log(1-\pt)\\ &=\mathbb{E}[\log (B_t^*(I_t))|L_{t-1}],
    \end{align*}
    with an equality iff $b_t=\pt$.
\end{proof}

Proposition~\ref{theo:log_optimal} shows that the log-optimal strategy is to choose the bets proportional to the probabilities $\pt$ and $1-\pt$. This strategy is also called Kelly gambling \citep{kelly1956new, breiman1961optimal}. Using the limiting permutation p-value $\plim$, we can also rewrite $\pt$ as
$$
\pt=\mathbb{P}(I_t=1|L_{t-1})=\mathbb{E}[\mathbb{P}(I_t=1|\plim)|L_{t-1}]=\mathbb{E}[\plim|L_{t-1}],$$
since $I_1,I_2,\dots$ are i.i.d. conditional on $\plim$ with $\mathbb{P}(I_t=1|\plim)=\plim$ (see Section~\ref{sec:resampling_risk}). Hence, the log-optimal strategy tries to learn the realized value of $\plim$ on our data set $X_0$ over time. If we would know the distribution of $\plim$, then we could apply the log-optimal strategy. For example, $\plim$ follows a uniform distribution on $[0,1]$ under $H_0$. Hence, $\plim|L_{t-1}=\ell$ follows a Beta distribution with parameters $\alpha=\ell+1$ and $\beta=t-\ell$ under $H_0$, which implies that $\pt=\mathbb{E}_{H_0}[\plim|L_{t-1}]=(L_{t-1}+1)/(t+1)$. Therefore, the log-optimal strategy under $H_0$ is the powerless passive betting strategy introduced in Section~\ref{sec:aggressive_strat}. This makes sense, as we should not be able gain wealth under the null. Under the alternative, $\plim$ is stochastically smaller than uniform. Therefore, $B_t^*(0)\geq B_t^*(1)$, meaning our wealth increases in case of a win (and decreases otherwise). Without additional information about the distribution of $\plim$ under $H_1$, it is not possible to pin down the log-optimal bet further.

In Appendix~\ref{appn:asymp_log_opt}, we characterize the asymptotic behavior of the log-optimal strategy which turns out to be quite different from the standard behavior observed in betting games. In particular, we prove that $B_t^*(I_t)\to 1$ almost surely for $t\to \infty$, implying that we can only accumulate wealth very slowly at late stages of the testing process.

\section{The binomial strategy}
\label{sec:our_strat}

We now use the oracle results from Section~\ref{sec:log-optimal} to define powerful (fully data-driven) strategies. Therefore, we choose $B_t(r)$ as the log-optimal strategy from \eqref{eq:log_opt_strategy}, where we replace $$\pt=\mathbb{P}(I_t=1|L_{t-1})=\mathbb{E}[\plim|L_{t-1}]$$ by some reasonable hyperparameter $p_t$. 
 Recall from the previous section that if we would know the distribution of $\plim$, the log-optimal bet would be to set $p_t=\mathbb{E}[\plim|L_{t-1}]$. However, such information is usually not available as permutation tests are used in non-parametric settings and the alternative $H_1$ is composite. An intuitive practical approach is therefore to choose some reasonable \enquote{working prior} density $f$ for $\plim$ instead and calculate the log-optimal betting strategy with respect to $f$, meaning to set $p_t=\mathbb{E}_{f}[\plim|L_{t-1}]$. We call this the \emph{mimicked log-optimal strategy}. Note that if $f$ is not the true density of $\plim$, then the mimicked log-optimal strategy is still valid (meaning that it results in valid e-values and p-values at any stopping time), since it is valid for all $p_t\in [0,1]$ due to \eqref{eq:cond_bet}. However, it would not be log-optimal in this case. 

The mimicked log-optimal strategy requires to calculate the posterior distribution $f_{\plim|L_{t-1}}$ of $\plim$ given the working prior $f$ and the previous number of losses $L_{t-1}$ to then set  $p_t=\mathbb{E}_{f_{\plim}|L_{t-1}}[\plim]$ at each step. This is computationally demanding and thus not well suited for our sequential permutation tests. For this reason, we begin with considering a constant working prior $\plim=\pchoice$ for some fixed $\pchoice \in [0,1]$ in this section, which does not require any updating. Note that this working prior never describes the true distribution of $\plim$, as it assumes that the observed test statistic $Y_0$ is fixed. However, it is a simple starting point and already bears interesting results. In particular, it yields a simple closed form for the wealth. In Section~\ref{sec:average_wealth}, we show how this can be used to derive closed forms for more complex working priors $f$, thus avoiding the computational burden of updating the posterior distribution. For interpretation, one can think of the working prior $\plim=\pchoice$ as our choice if we cheated and peeked at the actually realized value $y_0$ of $Y_0$ before choosing a betting strategy. In this case the only randomness comes from the random sampling mechanism.

\subsection{Closed-form wealth of the binomial strategy}

The binomial strategy is defined by \eqref{eq:log_opt_strategy} with some user-chosen constant $\pchoice \in [0,1]$ plugged in place of the unknown $\pt$. The wealth of the binomial strategy after $T$ permutations has a simple closed form, as we will show in the following.

\begin{proposition}
    Using the binomial strategy with parameter $\pchoice \in [0,1]$, after $T$ permutations and $\ell$ losses, the wealth equals \[W_T^\pchoice(\ell)=(T+1) \pchoice^\ell(1-\pchoice)^{T-\ell} {T \choose \ell},\]
    which does not depend on the order of the $\ell$ losses.
    \label{prop:wealth_log_optimal}
\end{proposition} 
\begin{proof}
    We need show that $W_T^{\pchoice}(\ell)=\prod_{t=1}^{T} B_t(I_t)=\pchoice^\ell(1-\pchoice)^{T-\ell}(T+1)!/(\ell!(T-\ell)!)$. The numerator follows immediately from the definition of $B_t(I_t)$. For the denominator, realize that if there were $s\geq 0$ losses before step $t$, it holds $L_{t-1}+1=s+1$, and if there were $r\geq 0$ wins, it holds $t-L_{t-1}=r+1$, independently from the number of previous wins and losses, respectively. Since we assumed that $\ell$ losses and $T-\ell$ wins happened, the above formula follows. 
\end{proof}

Note that $W_T^{\pchoice}(\ell)$ equals $T+1$ times the binomial likelihood of observing $\ell$ successes in $T$ trials of probability $\pchoice$. The name of our strategy was motivated by this fact. It is known that the binomial likelihood is maximised by the oracle choice of $\pmax=\ell/T$, yielding the next result.

\begin{proposition}
    The wealth $W_T^\pchoice(\ell)$ is maximized by the choice $\pchoice=\pmax=\ell/T$. We call this the oracle binomial strategy, as it cannot be specified in advance of seeing the data.  If $0<\ell<T$, the oracle binomial strategy achieves a wealth of $W_T^{\pmax}(\ell)\geq  T \sqrt{(T+1)/(\ell(T-\ell)2\pi e^{1/(6\ell)})}$, which grows linearly in $T$ for constant $\ell$, but is always $\geq \sqrt{2/(\pi e^{1/6})}\sqrt{T+1}$ for any $\ell$.\label{prop:max_wealth_p}
\end{proposition}

\begin{proof}
     We have, by Stirlings's approximation:
    \begin{align*}
        W_T^{\pmax}(\ell)&=(\pmax)^\ell(1-\pmax)^{T-\ell}(T+1)!/(\ell!(T-\ell)!) \\
        &= \left(\frac{\ell^\ell(T-\ell)^{T-\ell}}{T^T}\right) (T+1)!/(\ell!(T-\ell)!) \\
        &\geq \left(\frac{\ell^\ell(T-\ell)^{T-\ell}}{T^T}\right) \frac{1}{e}\frac{(T+1)^{T+1}}{\ell^\ell (T-\ell)^{T-\ell}}  \sqrt{\frac{T+1}{\ell(T-\ell) 2 \pi e^{1/(6\ell)}}} \\
        &\geq T \sqrt{\frac{T+1}{\ell(T-\ell) 2 \pi e^{1/(6\ell)}}},
    \end{align*}
    where we used $(1+(1/T))^{T}\geq e/(1+1/T)$ in the second inequality. The term $\ell(T-\ell)$ is minimized for $\ell=T/2$ and $e^{1/(6\ell)}$ for $\ell=1$ such that we obtain $W_T^{\pmax}(\ell)\geq   \sqrt{(T+1)2/(\pi e^{1/6})}$.
\end{proof}

Of course, $\pmax$ is not known to us in advance,  but if it were, then Proposition~\ref{prop:max_wealth_p} shows that we always end up with a wealth of at least $\sqrt{2/(\pi e^{1/6})}\sqrt{T+1}$, and it can be much larger if $\ell\ll T/2$ or $\ell\gg T/2$ (growing linearly in $T$ in that case). This might be surprising, as it also holds under $H_0$. However, note that even if we believe $H_0$ is true, we do not know $\pmax$ in advance. 

\subsection{A theoretically grounded choice for binomial parameter $p$}

Comparing the binomial strategy with the extreme ones in Section~\ref{sec:aggressive_strat}, the parameter $\pchoice$ can be interpreted as the aggressiveness with which we bet. The lower the $\pchoice$, the more aggressive we bet on the alternative. In the following, we introduce a choice of $\pchoice$ that adapts to the level $\alpha\in (0,1)$ the considered hypothesis is tested at and which offers a good compromise between aggressiveness and safety.


\begin{proposition}
    Let $N=\lceil{\sqrt{2 \pi e^{1/6}}/\alpha}\rceil$ and $\pchoice=1/N$. If $\pperm_T\leq \pchoice$ for some $T\in \mathbb{N}$, then there exists  $\tilde{T}\leq T$ such that  $W_{\tilde{T}}^{\pchoice}\geq  1/\alpha$. \label{prop:always_reject}
\end{proposition}

Before we prove Proposition~\ref{prop:always_reject}, we first we need the following lemma.

\begin{lemma}\label{lemma:strong_law}
    Let $\pchoice=1/N$ for some $N\in \mathbb{N}$, and let $\ell_N$ denote the realized number of losses after $N$ permutations. If $\pmax\leq \pchoice$ for some $T\geq N$ and $\ell_N\geq 1$, then there exists  $\tilde{T}\leq T$ such that $\pchoice=\ell_{\tilde{T}}/\tilde{T}$.
\end{lemma}
\begin{proof}[Proof of Lemma~\ref{lemma:strong_law}]
     Let $\tilde{T}=\min\{t\in \mathbb{N}: \ell_t/t\leq \pchoice, \ell_t\geq 1\}$.  Note that $\ell_{\tilde{T}}/(N\ell_{\tilde{T}})=\pchoice$. Therefore, $N\ell_{\tilde{T}}\leq \tilde{T}$ (otherwise $\ell_{\tilde{T}}/\tilde{T}>\pchoice$) and thus $\ell_{N\ell_{\tilde{T}}}\leq \ell_{\tilde{T}}$. Since $\ell_{N\ell_{\tilde{T}}}/(N\ell_{\tilde{T}})\leq \ell_{\tilde{T}}/(N\ell_{\tilde{T}}) = \pchoice$ and $\ell_{N\ell_{\tilde{T}}}\geq \ell_N \geq 1$, we have $\tilde{T}=N\ell_{\tilde{T}}=\ell_{\tilde{T}}/\pchoice$.
\end{proof}

\begin{proof}[Proof of Proposition~\ref{prop:always_reject}]
     First, assume that $\ell_N\geq 1$. Since $\ell_T/T\leq \pperm_T$ for all $T\in \mathbb{N}$, Lemma~\ref{lemma:strong_law} shows that there exists  $\tilde{T}\leq T$ such that $\pchoice=\ell_{\tilde{T}}/\tilde{T}=1/N$. Using the first bound of Proposition~\ref{prop:max_wealth_p}, we obtain \begin{align*}
        W_{\tilde{T}}^{\pchoice}(\ell_{\tilde{T}})&\geq \tilde{T} \sqrt{(\tilde{T}+1)/(\ell_{\tilde{T}}(\tilde{T}-\ell_{\tilde{T}})2\pi e^{1/(6\ell_{\tilde{T}})})} \\
        &= \sqrt{(\ell_{\tilde{T}}+\pchoice)/(\pchoice^2(1-\pchoice)2\pi e^{1/(6\ell_{\tilde{T}})})} \\
        &\geq \frac{1}{\pchoice\sqrt{2\pi e^{1/6}}} \quad \geq 1/\alpha.
    \end{align*}
    Now, consider the case $\ell_N=0$. The minimal $T\in \mathbb{N}$ with $\pperm_T\leq \pchoice$ is $T=N-1$.  With Proposition $\ref{prop:wealth_log_optimal}$ and $(1-1/N)^{N-1}\geq 1/e$, it follows 
    \begin{align*}
        W_{N-1}^{\pchoice}(0)=N\left(1-\frac{1}{N}\right)^{N-1}\geq \frac{N}{e} \geq \frac{\sqrt{2 \pi e^{\frac{1}{6}}}}{e\alpha} \geq \frac{1}{\alpha},
    \end{align*}
    since $\sqrt{2 \pi e^{1/6}}\geq e$. \end{proof}
        

Proposition~\ref{prop:always_reject} shows that for the defined $\pchoice$ the only possibility such that $\pperm_T$ rejects the hypothesis and the binomial strategy does not reject earlier is that $\alpha/\sqrt{2\pi e^{1/6}}<\pperm_T\leq \alpha$. Therefore, the factor $1/\sqrt{2\pi e^{1/6}}$ can be interpreted as loss we need to incorporate due to the sequential nature of our strategy. Note that some kind of loss is needed, as $\pperm_T$ is not valid for random stopping times; without the loss factor, Proposition~\ref{prop:always_reject} would imply that our strategy always rejects earlier than the non-valid $\pperm_T$. Since $\pperm_t|\{\plim=\ptrue\} \stackrel{a.s.}{\to} \ptrue $ for $t\to \infty$, Proposition \ref{prop:always_reject} immediately implies the following result.

\begin{corol}\label{rem:rr-binom}
   Let $\pchoice=1/\lceil{\sqrt{2 \pi e^{1/6}}/\alpha}\rceil$ and $\tau:=\inf\{t\geq 1: W_t^{\pchoice} \geq 1/\alpha\}$. If $\plim \in [0,\pchoice)$, then $\tau$ is almost surely finite and the test $\phi=\mathbbm{1}\{W_{\tau}^{\pchoice} \geq 1/\alpha\}$ has zero resampling risk ($RR_{\ptrue} (\phi)=0$). 
\end{corol}

 In Section \ref{sec:average_wealth} we propose a strategy that leads to zero resampling risk for all $\ptrue \in [0,c)$, where $c$ can be any predefined $c\in (0,\alpha)$. 
    

\subsection{Stopping when the wealth is low}
Currently, we assumed that as long as we do not reject the null hypothesis, we will keep generating additional permutations. However, we can also save a lot of permutations when we stop because it is very unlikely that we will reject the null hypothesis if we keep drawing more permutations. More explicitly, we suggest to stop when the current wealth drops below $\alpha$. In this case we would have a test martingale with value less than $\alpha$ and hence the conditional probability of then going on to reject the null hypothesis (crossing wealth $1/\alpha$) would be less than $\alpha^2$ under the null. Of course, the conditional probability can be much larger under the alternative, but in this case it is also much more unlikely that we end up with a wealth of $<\alpha$, so it is unlikely that this ``stopping out of futility'' will hurt power much --- we confirm this in simulations. Also note that we can slightly improve our procedure due to this stop. At each step $t$ we know our current wealth $W_t$. Therefore, we know whether or not we would stop if we lose in the next round --- and when this happens, we might as well bet everything on a win. For the binomial strategy this means that if $W_{t-1}\pchoice(t+1)/(L_{t-1}+1)<\alpha $, we choose $\pchoice=0$ for step $t$. Our final algorithm with the parameter $\pchoice$ as in Proposition~\ref{prop:always_reject} is described in Algorithm~\ref{alg:1}.

\begin{algorithm}
\caption{The binomial strategy, along with stopping for futility} \label{alg:1}
 \hspace*{\algorithmicindent} \textbf{Input:} Significance level $\alpha\in (0,1)$ and sequence of test statistics $Y_0,Y_1, Y_2, \ldots$ .\\
 \hspace*{\algorithmicindent} \textbf{Output:} Stopping time $\tau$ and wealth $W_1^\pchoice,\ldots,W_{\tau}^\pchoice$. 
\begin{algorithmic}[1]
\State $W_0^\pchoice = 1$
\State $L_0 = 0$ 
\State $\pchoice = 1/\lceil{\sqrt{2 \pi e^{1/6}}/\alpha}\rceil$
\For{$t=1,2,...$}
\If{$W_{t-1}^\pchoice\pchoice(t+1)/(L_{t-1}+1)<\alpha$} 
    \State $p_t=0$
\Else 
    \State $p_t=\pchoice$
\EndIf
\If{$Y_t \geq Y_0$} 
    \State $W_t^\pchoice = W_{t-1}^\pchoice p_t (t+1) /(L_{t-1}+1)$
    \State $L_{t}= L_{t-1}+1$
\Else 
    \State $W_t^\pchoice = W_{t-1}^\pchoice (1-p_t) (t+1) /(t-L_{t-1})$
    \State $L_{t}= L_{t-1}$
\EndIf
\If{$W_t^\pchoice<\alpha \textbf{ or } W_t^\pchoice\geq 1/\alpha$} 
\State $\tau = t$
\State \Return $\tau, W_1^\pchoice, \ldots, W_{\tau}^\pchoice$
\EndIf
\EndFor
\end{algorithmic}
\end{algorithm}

\section{The binomial mixture strategy\label{sec:average_wealth}}
As described in Section~\ref{sec:high_level}, the wealth process obtained by a betting strategy is a test martingale. The evidence contained in multiple dependent test martingales can be combined by averaging them (since averaging preserves the test martingale property). In our case, this means we can average the wealth obtained by various binomial strategies with different parameters $\pchoice$. In this section, we show how this can be used to find closed forms for the wealth of the mimicked log-optimal strategy for more complex working priors than $\plim=\pchoice$. 

\subsection{Connection to log-optimality}

Let $f$ be a density with respect to some probability measure $\mu$ on $([0,1],\mathcal{B})$, where $\mathcal{B}$ is the Borel $\sigma$-algebra. The binomial mixture strategy is defined by the wealth 
$$\bar{W}_T^f=\int_0^1 W_T^\pchoice f(\pchoice) \mu(d\pchoice),$$
where $W_T^\pchoice$ is the wealth of the binomial strategy after $T$ permutations. By Tonelli's Theorem, it can be easily verified that
$\bar{W}_t^f$ indeed defines a test martingale with respect to the filtration $(\mathcal{I}_t)_{t\geq 1}$ (informally, in words: averages of martingales are also martingales). It should be noted that while the individual betting strategies can (and will) depend on the past, the choice of strategies over which we average, in our case given by the density $f$, needs to remain the same over the entire testing process in order to ensure the test martingale property. 

\begin{theorem}\label{theo:log_opt_mixture}
The binomial mixture strategy with density $f$ is equivalent to the log-optimal strategy where $\pt=\mathbb{E}[\plim|L_{t-1}]$ is replaced by $p_t=\mathbb{E}_f[\plim|L_{t-1}]$. Hence, if $f$ is the true density of $\plim$, then the binomial mixture strategy is log-optimal.
\end{theorem}
\begin{proof}
  Let $W_T^{*,f}(\ell)$ be the wealth of the log-optimal strategy where $\pt$ is replaced by $p_t=\mathbb{E}_f[\plim|L_{t-1}]$ after $T$ permutations and $\ell$ losses and $L_{T}$ denote the random number of losses after $T$ permutations. We show that $W_T^{*,f}=\bar{W}_T^{f}$ by induction in $T$. For the initial step ($T=1$), simply note that $W_1^{*,f}(1)=2 \mathbb{E}_f[\plim]$ and $W_1^{*,f}(0)=2 (1-\mathbb{E}_f[\plim])$, while $\bar{W}_1^f(1)=\int_0^1 2 \pchoice f(p) \mu(d\pchoice) = 2 \mathbb{E}_f[\plim]$ and $\bar{W}_1^f(0)=\int_0^1 2 (1-\pchoice) f(\pchoice) \mu(d\pchoice) = 2 (1-\mathbb{E}_f[\plim])$. Next, assume the induction hypothesis (IH) that for some fixed but arbitrary $T\in \mathbb{N}$, it holds that $W_{T-1}^{*,f}(\ell)=\bar{W}_{T-1}^f(\ell)$ for all $\ell\in \{0,\ldots, T-1\} $. We will now prove the induction step ($T-1\to T$). First, we consider the case $I_T=1$. For $\ell\in \{1,\ldots,T\}$ and $W_{T-1}^{*,f}(\ell-1)= \bar{W}_{T-1}^f(\ell-1)\neq 0$, it holds
\begin{align*}
    W_T^{*,f}(\ell)&= W_{T-1}^{*,f}(\ell-1) \frac{T+1}{\ell}\mathbb{E}_f[\plim|L_{T-1}=\ell-1] \\
    &= W_{T-1}^{*,f}(\ell-1)\frac{T+1}{\ell} \int_0^1 \pchoice f_{\plim|L_{T-1}}(\pchoice|\ell-1) \mu(d\pchoice) \\
    &= W_{T-1}^{*,f}(\ell-1) \frac{T+1}{\ell}\frac{\int_0^1 \pchoice f_{L_{T-1}|\plim}(\ell-1|\pchoice) f(\pchoice) \mu(d\pchoice)}{\int_0^1 f_{L_{T-1}|\plim}(\ell-1|\pchoice) f(\pchoice) \mu(d\pchoice)} \\
    &=  W_{T-1}^{*,f}(\ell-1) \frac{T+1}{\ell} \frac{\frac{\ell}{T(T+1)} \int_0^1 W_T^\pchoice(\ell) f(\pchoice) \mu(d\pchoice) }{\frac{1}{T}\int_0^1 W_{T-1}^\pchoice(\ell-1) f(\pchoice) \mu(d\pchoice)} \\ &\stackrel{\text{IH}}{=}  W_{T-1}^{*,f}(\ell-1) \frac{ \bar{W}_T^f(\ell)}{W_{T-1}^{*,f}(\ell-1)} \\
    &= \bar{W}_T^f(\ell).
\end{align*}
 A similar calculation can be done for the case $I_T=0$ and $\ell\in \{0,\ldots,T-1\}$. This completes the proof of the induction step, and thus the overall proof.
\end{proof}
Theorem~\ref{theo:log_opt_mixture} shows that the mimicked log-optimal strategy for a given working prior can be obtained by a weighted average over simple binomial strategies. Note that the latter is often easier to compute, since we already derived a simple closed form for the wealth of binomial strategies in Proposition~\ref{prop:wealth_log_optimal}. Furthermore, it shows that properties of binomial strategies can be transferred to log-optimal strategies, e.g. that the wealth only depends on the number of losses and not the order the losses occur. On the other hand, Theorem~\ref{theo:log_opt_mixture} also shows that every weighted average of binomial strategies can be obtained by a single betting strategy in which the bets are defined as in \eqref{eq:log_opt_strategy} but  $\pt$ is replaced with $\mathbb{E}_f[\plim|L_{t-1}]$. Since the binomial mixture strategy with density $f$ is the mimicked log-optimal strategy based on working prior $f$ for $\plim$, we also refer to the density $f$ for the binomial mixture strategy as a ``working prior'' in the following.

\subsection{Uniform working prior: closed-form wealth\label{sec:closed_form_binomial_mixture}}

In practice, we typically do not know the true underlying distribution and thus it is difficult to determine the log-optimal working prior $f$. However, we might know for which values of $\plim$ we want our wealth to be large. Therefore, it makes sense to choose our working prior $f$ for $\plim$ accordingly, meaning to put most probability mass to these values. For example, if we follow the objective of minimizing the resampling risk (see Section~\ref{sec:resampling_risk}), then we want a large wealth if $\plim \leq \alpha$. Thus, a reasonable approach in this case is to pick $f$ to be uniform on $[0,\alpha]$. In the following theorem, we derive a simple closed form for the wealth obtained by the working prior 
\[
u_c(\pchoice)=\mathbbm{1}\{\pchoice\leq c\}/c,
\] 
where $0<c<1$. In this case, $\mu$ is the Lebesgue measure.

\begin{proposition}
     The wealth of the binomial mixture strategy with density $u_c$ after $T$ permutations and $\ell$ losses is given by $\bar{W}_T^{u_c}(\ell)= (1-\mathrm{Bin}(\ell;T+1,c))/c$, where $\mathrm{Bin}(\cdot ; T+1,c)$ denotes the cumulative distribution function of a binomial distribution with $T+1$ trials of probability $c$.
    \label{theo:wealth_ap}
\end{proposition}
\begin{proof}
    We show that 
    \begin{align}
        (T+1) {T \choose \ell} \int_0^c \pchoice^\ell (1-\pchoice)^{T-\ell} d\pchoice=   \mathrm{Bin}(T+1;T+1,c)-\mathrm{Bin}(\ell;T+1,c)\label{eq:ind_hyp}
    \end{align}
    for all $T\in \mathbb{N}$ and $\ell\in \{0,1,\ldots,T\}$ by (backward) induction in $\ell$. 

    For the initial step ($\ell=T$), simply note that $(T+1)\int_0^c \pchoice^T dp = c^{T+1}=\mathrm{Bin}(T+1;T+1,c)-\mathrm{Bin}(T;T+1,c)$.
    Next, assume the induction hypothesis (IH) that  \eqref{eq:ind_hyp} holds for some fixed but arbitrary $\ell\in \{1,\ldots,T\}$. We will now prove the induction step ($\ell\to \ell-1$). Using integration by parts, it follows that
    \begin{align*}
        &(T+1) {T\choose \ell-1} \int_0^c \pchoice^{\ell-1} (1-\pchoice)^{T-\ell+1} d\pchoice \\ &= (T+1) {T\choose \ell-1}\left( \frac{c^\ell (1-c)^{T-\ell+1}}{\ell}+ \frac{T-\ell+1}{\ell}\int_0^c \pchoice^{\ell} (1-\pchoice)^{T-\ell} d\pchoice \right) \\
        & \stackrel{\text{IH}}{=} {T+1\choose \ell} c^\ell (1-c)^{T-\ell+1} + \sum_{i=\ell+1}^{T+1} c^i (1-c)^{T-i+1} {T+1 \choose i} \\
        &= \sum_{i=\ell}^{T+1} c^i (1-c)^{T-i+1} {T+1 \choose i} .
    \end{align*}
    This completes the proof of the induction step, and thus the overall proof.
\end{proof}

With this, we can apply the binomial mixture strategy with working prior $u_c$ by just using the closed-form wealth at each step (see Algorithm~\ref{alg:bin_mix}). In the next subsection, we argue why we usually choose $c$ smaller than, but close to $\alpha$. One might think that a beta distribution with decreasing density is a better working prior for $\plim$ than $u_c$ and therefore should yield a higher power. In Appendix~\ref{sec:beta_dist}, we demonstrate via simulations that this is not the case.

\begin{algorithm}
\caption{The binomial mixture strategy with working prior $u_c$, along with stopping for futility} \label{alg:bin_mix}
 \hspace*{\algorithmicindent} \textbf{Input:} Significance level $\alpha\in (0,1)$, parameter $c<\alpha$ and sequence of test statistics $Y_0,Y_1, Y_2, \ldots$ .\\
 \hspace*{\algorithmicindent} \textbf{Output:} Stopping time $\tau$ and wealth $\bar{W}_1^{u_c},\ldots,\bar{W}_{\tau}^{u_c}$. 
\begin{algorithmic}[1]
\State $\bar{W}_0^{u_c} = 1$
\State $L_0 = 0$ 
\For{$t=1,2,...$}
\If{$Y_t \geq Y_0$} 
    \State $L_{t}= L_{t-1}+1$
\Else 
    \State $L_{t}= L_{t-1}$
\EndIf
\State $\bar{W}_t^{u_c}= (1-\mathrm{Bin}(L_{t};t+1,c))/c$
\If{$\bar{W}_t^{u_c}<\alpha \textbf{ or } \bar{W}_t^{u_c}\geq 1/\alpha$} 
\State $\tau = t$
\State \Return $\tau, \bar{W}_1^{u_c}, \ldots, \bar{W}_{\tau}^{u_c}$
\EndIf
\EndFor
\end{algorithmic}
\end{algorithm}


\subsection{Uniform working prior: zero resampling risk}

The wealth of the above binomial mixture strategy has many desirable properties. For example, in contrast to the simple binomial strategy, the wealth is decreasing for an increasing number of losses. In Section~\ref{sec:log-optimal}, we noted that this is actually the case for all working priors yielding a distribution stochastically smaller than a uniform on $[0,1]$. This monotonicity in the number of losses is shared with the classical permutation p-value and thus seems intuitive. In the following, we show that the binomial mixture strategy with working prior $u_c$ also shares some asymptotic behavior with the permutation p-value, which has implications for its resampling risk. 

\begin{theorem}
    The wealth of the binomial mixture strategy with uniform mixture density $u_c$ satisfies\begin{align*}\bar{W}_T^{u_c}|\{\plim=\ptrue\} \stackrel{a.s.}{\to}\begin{cases}
        1/c, &\text{ if } \ptrue \in [0,c) \\
        0, &\text{ if } \ptrue \in (c,1]
    \end{cases} \end{align*}   for $T\to \infty$.\label{theo:asymptotic_ap}
\end{theorem}
\begin{proof}
Let $\ptrue \in [0,c)$. By the De Moivre-Laplace Theorem the binomial distribution approaches a normal distribution, meaning for every $\epsilon >0$ there exists $T_1 \in \mathbb{N}$ such that $\mathrm{Bin}(\ell;T+1,c)\leq \Phi((\ell-Tc)/(\sqrt{Tc(1-c)}))+\epsilon $ for all $T\geq T_1$ and $\ell\in \{0,\ldots, T\}$, where $\Phi$ is the cumulative distribution function of a standard normal distribution. In addition, since $\sfrac{L_T}{T}|\{\plim=\ptrue\} \stackrel{a.s.}{\to} \ptrue$ for $T\to \infty$ and $\ptrue < c$, there exists a $T_2\in \mathbb{N}$ such that $L_T|\{\plim=\ptrue\}\leq T (\ptrue + \delta)$ a.s. for $T\geq T_2$ and some $\delta>0$ with $\delta+\ptrue <c$. Hence, \begin{align*} \mathrm{Bin}(L_T;T+1,c)|\{\plim=\ptrue\} \leq \Phi\left(\frac{\sqrt{T}(\ptrue+\delta -c)}{\sqrt{c(1-c)}}\right)+\epsilon   \end{align*}
for all $T\geq \max(T_1,T_2)$, which goes to $\epsilon$ for $T \to \infty$. Analogously, it follows that $\mathrm{Bin}(L_T;T+1,c)|\{\plim=\ptrue\} \stackrel{a.s.}{\to} 1$ if $\ptrue \in (c,1]$.
\end{proof}

Theorem~\ref{theo:asymptotic_ap} particularly shows that for $c=\alpha$, the wealth $\bar{W}_T^{u_\alpha}|\{\plim=\ptrue\} \stackrel{\text{a.s.}}{\to} 1/\alpha$ if $\ptrue<\alpha$ and  $\bar{W}_T^{u_\alpha}|\{\plim=\ptrue\}\stackrel{\text{a.s.}}{\to} 0$ if $\ptrue >\alpha$ for $T\to \infty$. Hence, if the event that the limiting permutation p-value equals $\alpha$ is a null set, meaning $\mathbb{P}(\plim=\alpha)=0$, then $\bar{W}_T^{u_\alpha}$ converges to the all-or-nothing bet \citep{shafer2021testing} on $\plim$, $\mathbbm{1}\{\plim \leq \alpha\}/\alpha$, almost surely. 
However, note that with the choice $c=\alpha$, Proposition~\ref{theo:wealth_ap} implies that $\bar{W}_T^{u_\alpha}<1/\alpha$ for all $T\in \mathbb{N}$ (meaning that if it converges to $1/\alpha$, this convergence occurs from below), and so we could never stop after a finite number of permutations. For this reason, when aiming for a level-$\alpha$ test in practice, one should choose $c$ slightly smaller than $\alpha$, which ensures a wealth of $\bar{W}_T^{u_c}>1/\alpha$ after a finite number of permutations if $\ptrue < c$. This is captured in the following corollary, which should be compared to Corollary~\ref{rem:rr-binom}.


\begin{corol}\label{corol:rr_mixture}
   Let $c<\alpha$ and  $\tau:=\inf\{t\geq 1: \bar{W}_t^{u_c} \geq 1/\alpha\}$. If $\plim \in [0,c)$, then $\tau$ is almost surely finite  and the test $\phi=\mathbbm{1}\{\bar{W}_{\tau}^{u_c} \geq 1/\alpha\}$ has zero resampling risk ($RR_{\ptrue} (\phi)=0$). 
\end{corol}

The parameter $c<\alpha$ can also be interpreted as the aggressiveness with which we bet on the alternative. If we choose $c$ very close to $\alpha$, we have the guarantee to always reject if $\ptrue \leq \alpha-\epsilon$ for an arbitrary small $\epsilon >0$. However, if the difference between $\alpha$ and $c$ is larger, we might be able to reject much earlier, which can be used to save computational cost. Thus, the binomial mixture strategy with density $u_c$ offers a simple trade-off between power and computation time.

The closed form of the wealth in Proposition~\ref{theo:wealth_ap} can easily be used to calculate or bound the required number of permutations for concrete parameter constellations. For example, if we observe no losses, we can reject the hypothesis as soon as $(1-(1-c)^{T+1})/c>1/\alpha$. Hence, if $\alpha=0.05$ and $c=0.04$, we can reject the null hypothesis after $39$ permutations and if $c=0.049$, we can reject after $77$ permutations. As another example, one can calculate that $\bar{W}_{200}^{u_{0.04}}(5)> 1/\alpha$. Thus, if we choose $c=0.04$, we can reject the null hypothesis at level $\alpha=0.05$ after $\leq 200$ permutations if we observe $\leq 5$ losses during these permutations.  

To summarize, the binomial mixture strategy with working prior $u_c$ allows to specify a parameter $c<\alpha$ arbitrarily close to $\alpha$ such that $H_0$ will be  rejected almost surely at some finite time if $\plim<c$. Note that this guarantee is not possible with the permutation p-value $\pperm$ nor Besag-Clifford p-value $\pbc$. No matter how large we choose the maximum number of permutations $T$ and the maximum number of losses $h$, there will always be a positive probability for accepting $H_0$. Furthermore, if there is sufficient evidence against $H_0$ early in the sampling process, the binomial mixture allows to stop early for rejection, while the permutation p-value and Besag-Clifford method never stop early for rejection. 

\subsection{Stochastic rounding of promising wealth\label{sec:stochastic_rounding}}
Suppose we stopped the testing process at some time $\tau$ where the current wealth looks promising ($W_{\tau}$ close to $1/\alpha$) but did not lead to a rejection yet ($\max_{s=1,\ldots,\tau} W_{s}<1/\alpha$). For example, this might be when a prespecified maximum number of permutations is hit or the wealth changed extremely slowly during the last permutations. As described in Section \ref{sec:high_level}, we could continue sampling or interpret the wealth on its own, since a large wealth is still good evidence against the null hypothesis. If both these options fall out, as drawing more permutations is not possible or wanted and we are only interested in rejections, we could also use a technique called stochastic rounding to possibly achieve a rejection~\citep{xu2023more}.

This is based on a randomized improvement of Ville's inequality by \citet{ramdas2023randomized}. They showed that for any stopping time $\tau$ with respect to the filtration $(\mathcal{I}_t)_{t\geq 1}$, it holds that
$$\mathbb{P}(W_{\tau} \geq U/\alpha)\leq \alpha,$$
where $\alpha\in (0,1)$ and $U\sim U[0,1]$ is independent of the filtration $(\mathcal{I}_t)_{t\geq 1}$. Thus, if we stopped at time $\tau$ before reaching the level $1/\alpha$, we can draw a sample $u$ from a uniformly distributed random variable $U\sim U[0,1]$ and reject the null hypothesis if $W_\tau\geq u/\alpha$. However, in order to maintain validity, it is important that $u$ is sampled after stopping the process, such that $\tau$ and $W_\tau$ are independent of $U$, and that it is only sampled once (and not multiple times to pick the best one). Therefore, this technique should only be used in an automated code which prevents this randomized improvement from being exploited for cheating. The term stochastic rounding stems from this test being equivalent to rejecting if $W_{\tau}^{\sim}\geq 1/\alpha$, where $W_{\tau}^{\sim}$ is an e-value defined by $W_{\tau}^{\sim}=W_{\tau}$ if $W_{\tau}\geq 1/\alpha$; and $W_{\tau}^{\sim}=1/\alpha$ with probability $W_\tau \alpha$ and $W_{\tau}^{\sim}=0$ with probability $(1-W_\tau) \alpha$ if $W_{\tau}< 1/\alpha$.

For example, stochastic rounding could be used to set the resampling risk for all $\ptrue<\alpha$ to an arbitrary small $\epsilon\in (0,1)$ based on the binomial mixture strategy with working prior $u^{\alpha}$, which follows immediately by Theorem \ref{theo:asymptotic_ap}. 


\begin{corol}\label{corol:rr-random_mixture}
    Let $\epsilon \in (0,1)$ and $\tau=\inf\{t\geq 1: \bar{W}_t^{u_{\alpha}}\geq (1-\epsilon)/\alpha \}$. If $\plim \in [0,\alpha)$, then $\tau$ is almost surely finite  and the randomized test $\phi=\mathbbm{1}\{\bar{W}_{\tau}^{u_{\alpha}}\geq (1-\epsilon)/\alpha, U\leq (1-\epsilon) \}$ has resampling risk $\epsilon$ ($\mathrm{RR}_{\ptrue}(\phi)=\epsilon$). 
\end{corol}

\section{Every permutation test can be recovered by betting\label{sec:betting_is_general}}

A common concern regarding sequential testing and the use of e-values is that power is lost compared to classical p-value  methods like the permutation p-value $\pperm$. In this section, we show that every permutation test that is a function of the number of losses $L_t$ after a fixed number of permutations $t$ can be obtained by our strategy. In particular, we demonstrate that the permutation p-value and the Besag-Clifford method arise naturally as concrete instances of our general algorithm (Algorithm \ref{alg:general}) by prespecifying the number of permutations or the number of losses, respectively. 

\subsection{Recursive betting strategy for the construction of any permutation e-value}
Consider an e-value $E_t(L_t)$ that is defined by a non-random function $E_t$ of the number of losses $L_t$ for some prespecified $t\geq 1$. Similarly to $B_t$, $E_t$ can be interpreted as a betting function where $E_t(\ell)$ is our realized wealth if  $L_t=\ell$. Instead of betting sequentially at each step, the e-value $E_t(L_t)$ is just obtained by betting once on the number of losses after $t$ steps. Under $H_0$, the number of losses $L_{t}$ are distributed uniformly on $0,1,\dots, t$. Hence, in order to define an (admissible) e-value the betting function $E_t$ needs to satisfy
\begin{align}\sum_{\ell=0}^{t} E_t(\ell)=t+1.\label{eq:cond_bet_e-value}\end{align}
We first show that every $E_t$ can also be obtained by our sequential betting strategy --- this implies that we do not lose something due to the sequential betting. For this, note that after betting with $E_{t-1}$ on $L_{t-1}$, we could also draw a further permuted test statistic $Y_t$, bet on the indicator $I_t=\mathbbm{1}\{Y_t\geq Y_0\}$ with $B_t$ and define a new e-value as $E_{t-1}(L_{t-1}) \cdot B_t(I_t)$. 

We claim that every e-value $E_{t}(L_t)$ can be obtained by such a product.
Indeed, given an arbitrary betting vector $E_{t}$, and letting $B_{t|\ell}(1)$ (or $B_{t|\ell}(0)$) denote our bet on a loss (or win) at step $t$ if $L_{t-1}=\ell$, consider the following system of equations:
\begin{align}
     E_{t-1}(\ell) \cdot B_{t|\ell}(1)&=E_{t}(\ell+1) \quad  & \text{ for all } \ell=0,\ldots,t-1\label{eq:1}\\
    E_{t-1}(\ell) \cdot B_{t|\ell}(0)&=E_{t}(\ell) \quad  & \text{ for all } \ell=0,\ldots,t-1 \label{eq:2}\\
    B_{t|\ell}(0)\frac{t-\ell}{t+1}+B_{t|\ell}(1)\frac{\ell+1}{t+1} &=1 \quad  & \text{ for all } \ell=0,\ldots,t-1 \label{eq:3}\\ 
    E_{t-1}(0)+ \ldots + E_{t-1}(t-1) &=t.\label{eq:4}
\end{align} 
The equations \eqref{eq:1} and \eqref{eq:2} ensure that $E_{t-1} \cdot B_t$ leads to the same e-value as $E_{t}$ and the equations \eqref{eq:3} and \eqref{eq:4} ensure that $B_t$ and $E_{t-1}$ are valid betting functions according to \eqref{eq:cond_bet} and \eqref{eq:cond_bet_e-value}, respectively.

\begin{theorem}\label{theo:inc_gen_eval}
    Every e-value $E_{t}(L_t)$ can be obtained by a sequential betting strategy $E_{t}(L_t)=\prod_{r=1}^{t} B_{r|L_{r-1}}(I_r)$, where $B_{r|\ell}(0)$ and $B_{r|\ell}(1)$ are defined recursively via 
\begin{align*}
    E_{t-1}(\ell)&=\frac{\ell+1}{t+1} E_{t}(\ell+1)+ \frac{t-\ell}{t+1} E_{t}(\ell) \\
    B_{t|\ell}(0)&=\frac{E_{t}(\ell)}{E_{t-1}(\ell)} \\
    B_{t|\ell}(1)&=\frac{E_{t}(\ell+1)}{E_{t-1}(\ell)},
\end{align*}
which are the unique solutions to the equations \eqref{eq:1}-\eqref{eq:4}. We use the convention $0/0=0$.
\end{theorem}
\begin{proof}
    We solve for equations \eqref{eq:1}-\eqref{eq:3} and then show that \eqref{eq:4} is automatically fulfilled as well. Equations \eqref{eq:2} and \eqref{eq:3} immediately give $B_{t|\ell}(0)=\frac{E_{t}(\ell)}{E_{t-1}(\ell)}$ and $ B_{t|\ell}(1)=\frac{t+1-(t-\ell)B_{t|\ell}(0)}{\ell+1}$. Inserting this into \eqref{eq:1} yields
    \begin{align*}
        &  E_{t-1}(\ell) \left( \frac{t+1-(t-\ell)\frac{E_{t}(\ell)}{E_{t-1}(\ell)}}{\ell+1}\right)=E_{t}(\ell+1) \\
         \Leftrightarrow \quad & E_{t-1}(\ell)=\frac{\ell+1}{t+1} E_{t}(\ell+1)+ \frac{t-\ell}{t+1} E_{t}(\ell). 
    \end{align*}
    Furthermore, note that 
    \begin{align*} \sum_{\ell=0}^{t-1} E_{t-1}(\ell)&= \sum_{\ell=0}^{t-1} \frac{\ell+1}{t+1} E_{t}(\ell+1) + \sum_{\ell=0}^{t-1} \frac{t-\ell}{t+1} E_{t}(\ell) \\
    &= \frac{t}{t+1} E_{t}(t)+ \frac{t}{t+1} E_{t}(0)+ \sum_{\ell=1}^{t-1} \frac{\ell}{t+1} E_{t}(\ell) + \sum_{\ell=1}^{t-1} \frac{t-\ell}{t+1} E_{t}(\ell) \\
    &= \frac{t}{t+1} E_{t}(t)+ \frac{t}{t+1} E_{t}(0)+ \frac{t}{t+1} \sum_{\ell=1}^{t-1}  E_{t}(\ell) \\
    &= \frac{t}{t+1} \sum_{\ell=0}^{t}  E_{t}(\ell)  = t,\end{align*}
    which proves that~\eqref{eq:4} is fulfilled. The second claim follows by backward induction in $t$.
\end{proof}
This theorem shows that every e-value that is a function of $L_t$ for a fixed $t$ can also be obtained by our sequential betting strategy. Therefore, also every level-$\alpha$ test $\phi_{\alpha}(L_t)$, where $\phi_{\alpha}(L_t)$ takes values in $\{0,1\}$ and $\mathbb{P}_{H_0}(\phi_{\alpha}(L_t)=1)\leq \alpha$, can be obtained (or uniformly improved) by our sequential betting strategy, since $E_t^{\alpha}(L_t)=\phi_{\alpha}(L_t)/\alpha$ is a valid e-value. This is captured in the following corollary.

\begin{corol}
    For every level-$\alpha$ test $\phi_{\alpha}(L_t)$ that only depends on the number of losses $L_t$ for some $t\geq 1$, there is a betting strategy $B_1,\ldots,B_t$ such that $\mathbbm{1}\{W_t\geq 1/\alpha\}\geq \phi_{\alpha}(L_t)$ almost surely, where $W_t=\prod_{r=1}^t B_r(I_s)$. 
\end{corol}
Note that every p-value $\pval$ gives rise to a family of level-$\alpha$ tests $\boldsymbol{\phi}^{\pval}=(\phi_{\alpha}^{\pval})_{\alpha\in [0,1]}$ by $\phi_{\alpha}^{\pval}=\mathbbm{1}\{\pval\leq \alpha\}$. Thus, if each of these $\phi_{\alpha}^{\pval}$ depends only on $L_t$ for some $t\geq 1$, we could derive sequential versions of these level-$\alpha$ tests and then calibrate them into an anytime-valid p-value. In the following subsections, we use this to obtain anytime-valid generalizations of the permutation p-value \eqref{eq:perm-pval} and the Besag-Clifford p-value \eqref{eq:bc_pval}.

\subsection{Generalizing the permutation p-value by betting}

Let $T\in \mathbb{N}$ and $\alpha \in [0,1]$ be fixed and consider the e-value $E_T^{\alpha} (L_T)$ with the betting function $E_T^{\alpha}$ given by
\begin{align}
E_T^{\alpha}(\ell)=
\begin{cases}
            1/\alpha, &  \ell\leq \lfloor (T+1)\alpha \rfloor-1 \\
            a, & \ell=\lfloor (T+1)\alpha \rfloor \\ 
            0, & \text{otherwise}
        \end{cases} \quad (\ell\in \{0,1,\ldots,T\}),
\label{eq:e-value_pperm}
\end{align}
 where $a$ is a constant $0\leq a <1/\alpha$ such that the sum of $E_T^{\alpha}(\ell)$ equals $T+1$. If $(T+1)\alpha$ is a whole number, $a$ equals $0$. Theorem \ref{theo:inc_gen_eval} shows that this e-value can be obtained by our sequential betting strategy $E_T^{\alpha}(L_T)=\prod_{t=1}^T B_t^{\alpha,T}(I_t)$ for some betting functions $B_1^{\alpha,T},\ldots,B_T^{\alpha,T}$.  Define 
\begin{align}
\pval_{t}^{\mathrm{av},T}:=\inf\left\{\alpha \in [0,1]\Bigg\vert\exists r=1,\ldots,t: \prod_{s=1}^{\min(r,T)} B_s^{\alpha,T}(I_s) \geq 1/\alpha\right\}. \label{eq:general_pperm}
\end{align}
In the following proposition, we show that $(\pval_{t}^{\mathrm{av},T})_{t\in \mathbb{N}}$ is an anytime-valid generalization of the permutation p-value \eqref{eq:perm-pval}.

\begin{proposition}\label{prop:general_perm_pval}
    $(\pval_{t}^{\mathrm{av},T})_{t\in \mathbb{N}}$ is an anytime-valid p-value with $$\pval_{\tau}^{\mathrm{av},T}=(L_{\tau}+1+T-\tau)/(T+1)$$ for all stopping times $\tau\leq T$ and $\pval_{\tau}^{\mathrm{av},T}=\pval_{T}^{\mathrm{av},T}$ otherwise. In particular, $\pval_{T}^{\mathrm{av},T}=\pperm_T$, implying that $(\pval_{t}^{\mathrm{av},T})_{t\in \mathbb{N}}$ is an anytime-valid generalization of the permutation p-value \eqref{eq:perm-pval}.
\end{proposition}
\begin{proof}
   Ville's inequality implies that $(\pval_{t}^{\mathrm{av},T})_{t\in \mathbb{N}}$ is an anytime-valid p-value. Furthermore, note that for all $\tau\leq T$ we have $\prod_{s=1}^\tau B_s^{\alpha,T}(I_s)=E_{\tau}^{\alpha,T}(L_{\tau})$, where $E_{\tau}^{\alpha,T}$ can be calculated recursively by $E_{t-1}^{\alpha,T}(\ell)=\frac{\ell+1}{t+1} E_{t}^{\alpha,T}(\ell+1)+ \frac{t-\ell}{t+1} E_{t}^{\alpha,T}(\ell)$ with the starting point $E_{T}^{\alpha,T}=E_T^{\alpha}$ (see Theorem \ref{theo:inc_gen_eval}). Since $E_T^{\alpha}(\ell)=1/\alpha $ for all $\ell\leq \lfloor (T+1)\alpha \rfloor-1$ and $E_T^{\alpha}(\ell)<1/\alpha $ otherwise, this shows that $E_{\tau}^{\alpha,T}$ is given by 
$$
E_{\tau}^{\alpha,T}(\ell)=
\begin{cases}
            1/\alpha, &  \ell\leq \lfloor (T+1)\alpha \rfloor-1-(T-\tau) \\
            a_\ell, & \lfloor(T+1)\alpha \rfloor-(T-\tau)\leq \ell \leq \lfloor (T+1)\alpha \rfloor \\ 
            0, & \text{otherwise}
        \end{cases} \quad (\ell\in \{0,1,\ldots,T\}), 
$$
where the $a_\ell$ are constants $0\leq a_\ell <1/\alpha$. Hence, $E_{\tau}^{\alpha,T}(L_{\tau}) \geq 1/\alpha$ iff $L_{\tau}+1 \leq (T+1)\alpha - (T-\tau)$, which shows that $\pval_{\tau}^{\mathrm{av},T}=(L_{\tau}+1 + T-\tau)/(T+1)$ for $\tau\leq T$.
\end{proof}

Note that $\pval_{\tau}^{\mathrm{av},T}$ is the permutation p-value $\pperm_T$ we would obtain if all permutations after step $\tau$ were losses.

\subsection{Generalizing the Besag-Clifford p-value by betting}

The Besag-Clifford p-value $\pbc$ \eqref{eq:bc_pval} can be generalized in the same way as the permutation p-value. For this, consider again the betting function $E_{T}^{\alpha}$ defined in \eqref{eq:e-value_pperm} but let $T=\min\left(T_{\max},\lceil h /\alpha\rceil-1\right)$, where $h$ and $T_{\max}$ are the predefined number of losses and number of permutations for the Besag-Clifford method, and define  
\begin{align}
\pval_{t}^{\mathrm{av},T_{\max},h}:=\inf\left\{\alpha \in [0,1]\Bigg\vert\exists r=1,\ldots,t: \prod_{s=1}^{\min\left(r,T_{\max},\lceil h /\alpha\rceil-1\right)} B_s^{\alpha,\min\left(T_{\max},\lceil h /\alpha\rceil-1\right)}(I_s) \geq 1/\alpha\right\}.  \label{eq:general_BC}
\end{align}

\begin{proposition}
    $(\pval_{t}^{\mathrm{av},T_{\max},h})_{t\in \mathbb{N}}$ is an anytime-valid p-value with $$\pval_{\tau}^{\mathrm{av},T_{\max},h}=\min\left(\frac{h}{\tau+h-L_{\tau}}, \frac{L_{\tau}+1+T_{\max}-\tau}{T_{\max}+1}\right)$$ for all stopping times $\tau\leq \gamma(h,T_{\max})$ and $\pval_{\tau}^{\mathrm{av},T_{\max},h}=\pval_{\gamma(h,T_{\max})}^{\mathrm{av},T_{\max},h}$ otherwise. In particular, $\pval_{\gamma(h,T_{\max})}^{\mathrm{av},T_{\max},h}=\pbc_{\gamma(h,T_{\max})}$, implying that $(\pval_{t}^{\mathrm{av},T})_{t\in \mathbb{N}}$ is an anytime-valid generalization of the Besag-Clifford p-value \eqref{eq:bc_pval}.
\end{proposition}
\begin{proof}
Ville's inequality implies that $(\pval_{t}^{\mathrm{av},T_{\max},h})_{t\in \mathbb{N}}$ is an anytime-valid p-value. Furthermore, let $\tau\leq \gamma(h,T_{\max})$ be any stopping time. First, consider $\alpha\in [0,1]$ such that $\lceil h/\alpha \rceil -1<\tau$, implying that $L_{\lceil h/\alpha \rceil -1}\leq h-1$ and $T_{\max}>\lceil h /\alpha\rceil-1$. Then, we have $E_{\min\left(T_{\max},\lceil h /\alpha\rceil-1\right)}^{\alpha}(L_{\lceil h/\alpha \rceil -1})=1/\alpha$. Now consider $\alpha \in [0,1]$ such that $\tau \leq \min(T_{\max},\lceil h/\alpha \rceil -1)$. With the same reasoning as in the proof of Proposition~\ref{prop:general_perm_pval}, it follows that $E_{\tau}^{\alpha,\min(T_{\max},\lceil h/\alpha \rceil -1)}(L_{\tau})\geq 1/\alpha$ iff $$L_{\tau}+1\leq (\min(T_{\max},\lceil h/\alpha \rceil -1) +1) \alpha - (\min(T_{\max},\lceil h/\alpha \rceil -1)-\tau).$$ If $T_{\max}< \lceil h/\alpha \rceil -1$, this is equivalent to $(L_{\tau}+1+T_{\max}-\tau)/(T_{\max}+1)\leq \alpha$ and if $T_{\max}\geq \lceil h/\alpha \rceil -1$, it is equivalent to $h/(\tau+h-L_{\tau}))\leq \alpha$. Altogether, this implies the assertion. 
\end{proof}
Again, $\pval_{\tau}^{\mathrm{av},T_{\max},h}$ is the same as the Besag-Clifford $p$-value $\pbc_{\gamma(h,T_{\max})}$ if we would only observe losses after step $\tau$; and if we would continue sampling after $\gamma(h,T_{\max})$, the p-value $\pval_{\tau}^{\mathrm{av},T_{\max},h}$ would not change.
This generalizes our results from Section \ref{sec:besag_clifford}, where we showed that the aggressive strategy recovers the Besag-Clifford method for $h=1$, to general $h$. Indeed, $\pval_{t}^{\mathrm{av},\infty,1}=\pagg_t$ for all $t\in \mathbb{N}$.

It is interesting how the anytime-valid versions of the permutation p-value $\pperm$ and the Besag-Clifford p-value $\pbc$ arise from our approach. Indeed, the betting function $E^T_{\alpha}$ in \eqref{eq:e-value_pperm} is the most powerful bet we could make for fixed $\alpha$ and number of permutations $T$. Therefore, the resulting sequential betting strategy $B_1^{\alpha,T},\ldots, B_T^{\alpha,T}$ is optimal if $\alpha$ and $T$ are fixed. However, we claim that fixing the number of permutations is not the best way to perform permutation testing, since it is more reasonable to draw a larger number of permutations if the decision is tight and fewer permutations if the decision is unambiguous. This is provided by our mimicked log-optimal strategies.



\section{Experiments on simulated and real data\label{sec:simulations}}

We begin with three subsections on simulations, and end with two subsections on real data. In the first three subsections, we simulated $m=2000$ independent treatment vs.\ control trials (see Example \ref{example:Fisher}). For each trial, we generated $n=1000$ observations, where the probability that an observation was treated is $0.5$. Responses of control observations were generated from $N(0,1)$ and responses of treated observations from $N(\mu,1)$. We test the null hypothesis of independence between the treatment and response. For  $Y_t$, $t\geq 0$, we chose the difference in mean between the treated and untreated observations and $T=1000$ permutations were drawn randomly with replacement.

\subsection{Behavior of the anytime-valid p-values\label{sec:sim_p-value}}
In Figure~\ref{fig:sim_results} we compare the (log-transformed) p-values in ascending order obtained by the classical permutation p-value $\pperm_T$ with the sequential p-values obtained from the binomial mixture strategy, the binomial strategy and the aggressive strategy. The p-values of the sequential strategies are defined by $1/\max_{b=1,\ldots,T} W_b$. Due to Ville's inequality, these are valid p-values. The binomial mixture strategy was applied with density $u_c$, $c=0.9 \alpha$, and the parameter of the binomial strategy was chosen as in Algorithm~\ref{alg:1}.

\begin{figure}[h!]
\centering
\includegraphics[width=15cm]{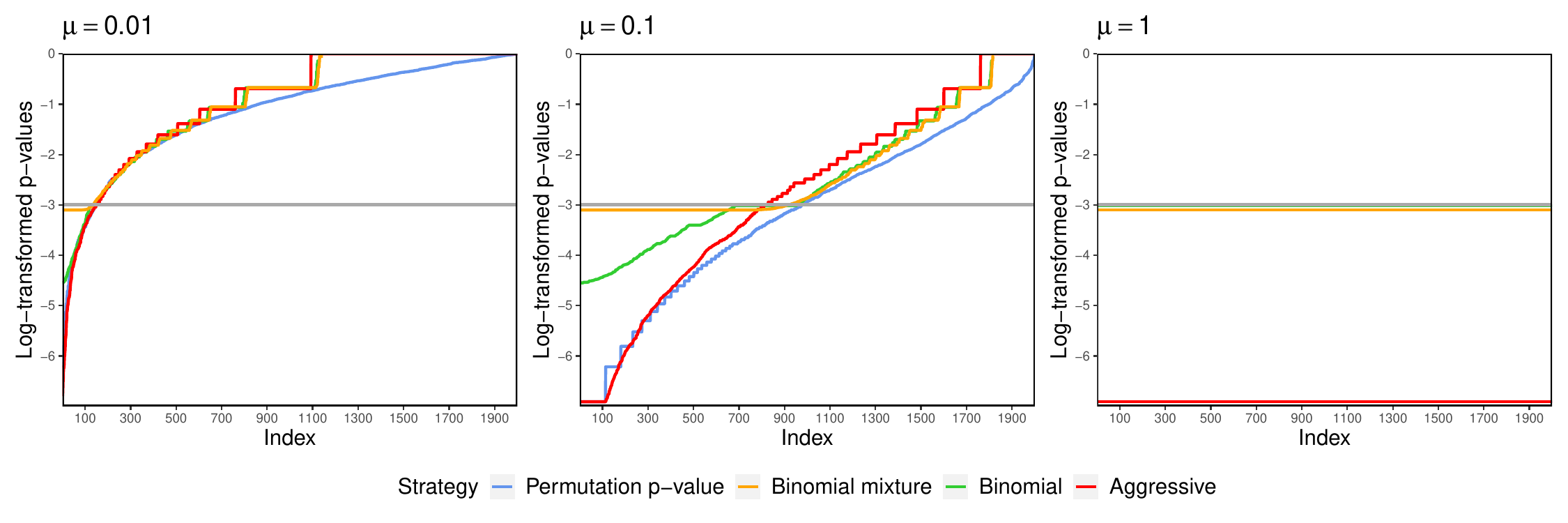}\\
\includegraphics[width=15cm]{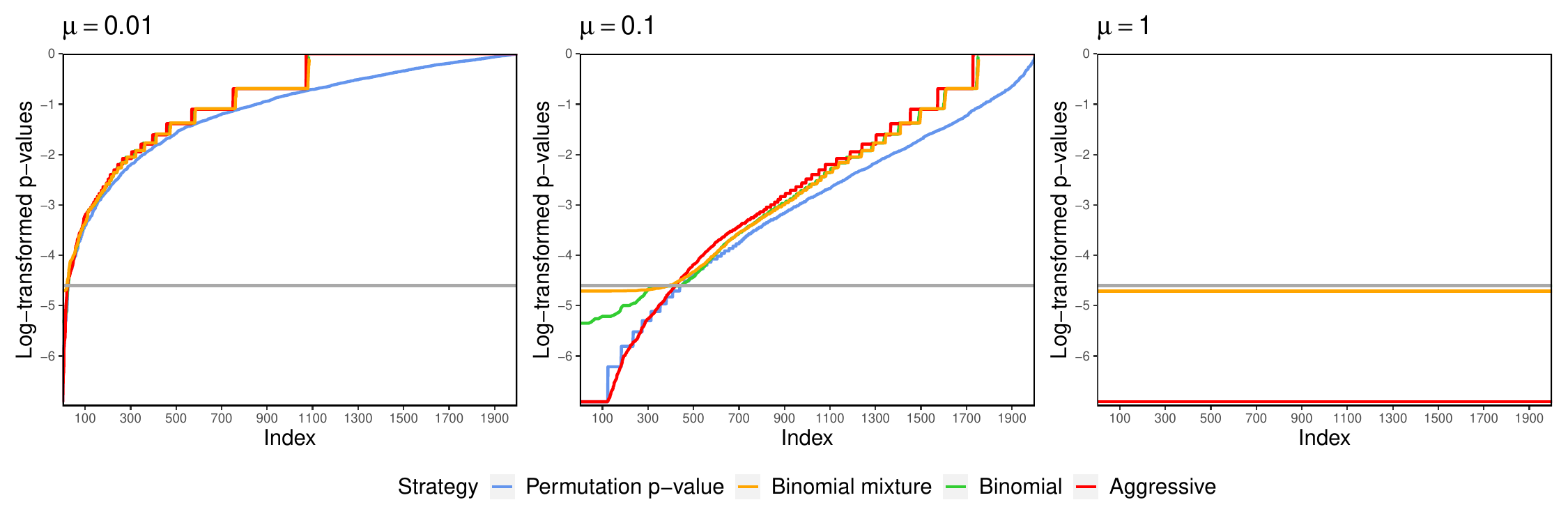}%
\caption{All plots display log-transformed p-values in ascending order for $2000$ simulations --- following the experimental protocol in Subsection~\ref{sec:sim_p-value}. The grey horizontal line equals $\log(\alpha)$, where $\alpha$ equals 0.05 (top row) or 0.01 (bottom row). Lower is better, so the permutation p-value performs best, but when focusing on the number of p-values less than $\alpha$ (the power at level $\alpha$), it is similar to our betting strategies. In the right plots, the permutation p-value and the aggressive strategy p-value are the same. \label{fig:sim_results} }\end{figure}

The results show that the permutation p-value performs best in all cases, except for $\mu=1$, where the aggressive strategy and the permutation p-value are the same. Nevertheless, the number of log-p-values below the level $\log(\alpha)$, which is illustrated by the grey horizontal line, is nearly the same for the binomial mixture strategy, binomial strategy and the permutation p-value. Therefore, if we test at level $\alpha$, the power should be quite similar, while the number of permutations is lower using the sequential strategies (we will examine the power in detail in the subsequent sections). Further, note that $T$ is fixed for the permutation p-value in advance, while we could decide to carry on testing using the sequential strategies in cases where the current wealth looks promising.


\subsection{Power and sample size comparison with the Besag-Clifford method\label{sec:sim_power}}

In Figure~\ref{fig:sim_results_power} and~\ref{fig:sim_results_power_2}, we explore the power and number of permutations needed to obtain the decision for $\alpha=0.05$ and $\alpha=0.01$. The power is displayed in the upper left plots, the average number of permutations in the upper right plots and the average number of permutations when the testing process was stopped for rejection and futility in the lower left and right plots, respectively. The sequential strategies are stopped for futility if the wealth drops below $\alpha$ and for rejection if the wealth exceeds $1/\alpha$. We replaced the classical permutation p-value $\pperm_T$ with the Besag-Clifford sequential strategy \citep{besag1991sequential} with parameter $h=\alpha T$. In this case, the power of $\pbc_{\gamma(h,T)}$ is the same as the power of $\pperm_{T-1}$, even though the Besag-Clifford approach may stop before the maximum number of permutations $T$ are sampled \citep{silva2009power}. Our sequential strategies are applied with the same parameters as described in Section~\ref{sec:sim_p-value}.

\begin{figure}[h!]
\centering
\includegraphics[width=15cm]{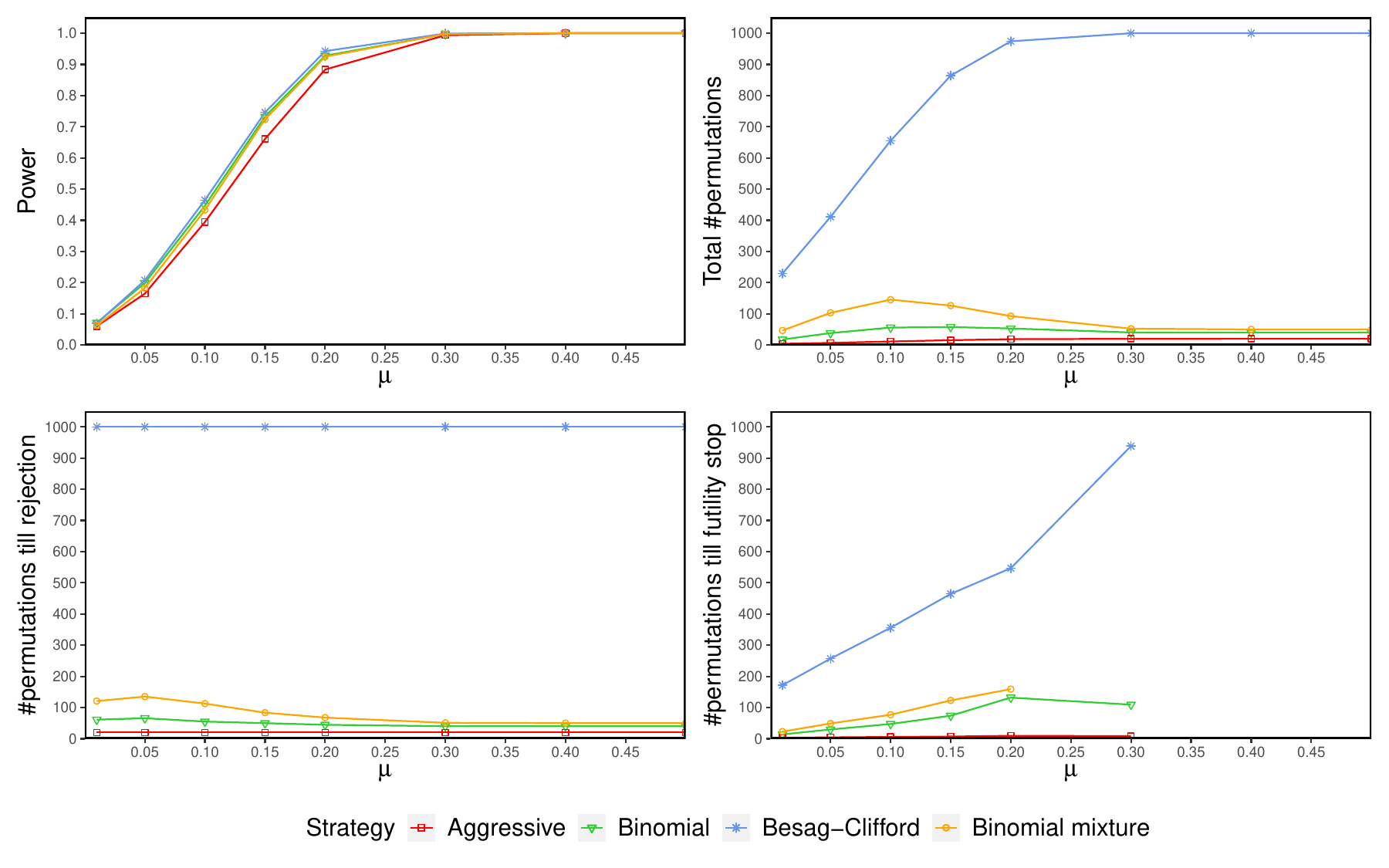}%
\caption{Power and average number of permutations until the decision was obtained for  $\alpha=0.05$ and different strengths of the alternative $\mu$ --- following the experimental protocol in Subsection~\ref{sec:sim_power}. Upper left plot: Empirical power; Upper right plot: Average number of permutations $\bar{\tau}$; Lower left plot: Average number of permutations until it was stopped for rejection $\bar{\tau}_1$; Lower right plot: Average number of permutations until it was stopped for futility $\bar{\tau}_0$. Relationship between the plots is given by $\bar{\tau}=(\bar{\tau_1}m_1+\bar{\tau}_0m_0+T(m-m_1-m_0))/m$, where $m_1$ and $m_0$ are the number of simulation runs where it was stopped for rejection and futility, respectively. The power of binomial strategy and binomial mixture strategy is similar as with Besag-Clifford and number of permutations is reduced substantially, particularly when the alternative is strong. \label{fig:sim_results_power} }\end{figure}

In line with Figure~\ref{fig:sim_results}, the power obtained by the binomial strategy and binomial mixture strategy is only slightly worse than the one obtained by Besag-Clifford. However, the number of permutations until a decision is obtained can be reduced a lot by our proposals. In particular, note that the Besag-Clifford strategy only stops earlier for futility, but never in case of a rejection (lower left plots). Therefore, it works well under under the null hypothesis or under very weak alternatives, but under strong alternatives it does not save any computational time compared to the classical approach. In case of $\alpha=0.01$ (Figure \ref{fig:sim_results_power_2}), the computational gain of our strategies is lower and the power loss of the binomial mixture strategy larger. In the next section, we show that this power loss can be avoided using the stochastic rounding technique described in Section \ref{sec:stochastic_rounding}. 


\begin{figure}
\centering
\includegraphics[width=15cm]{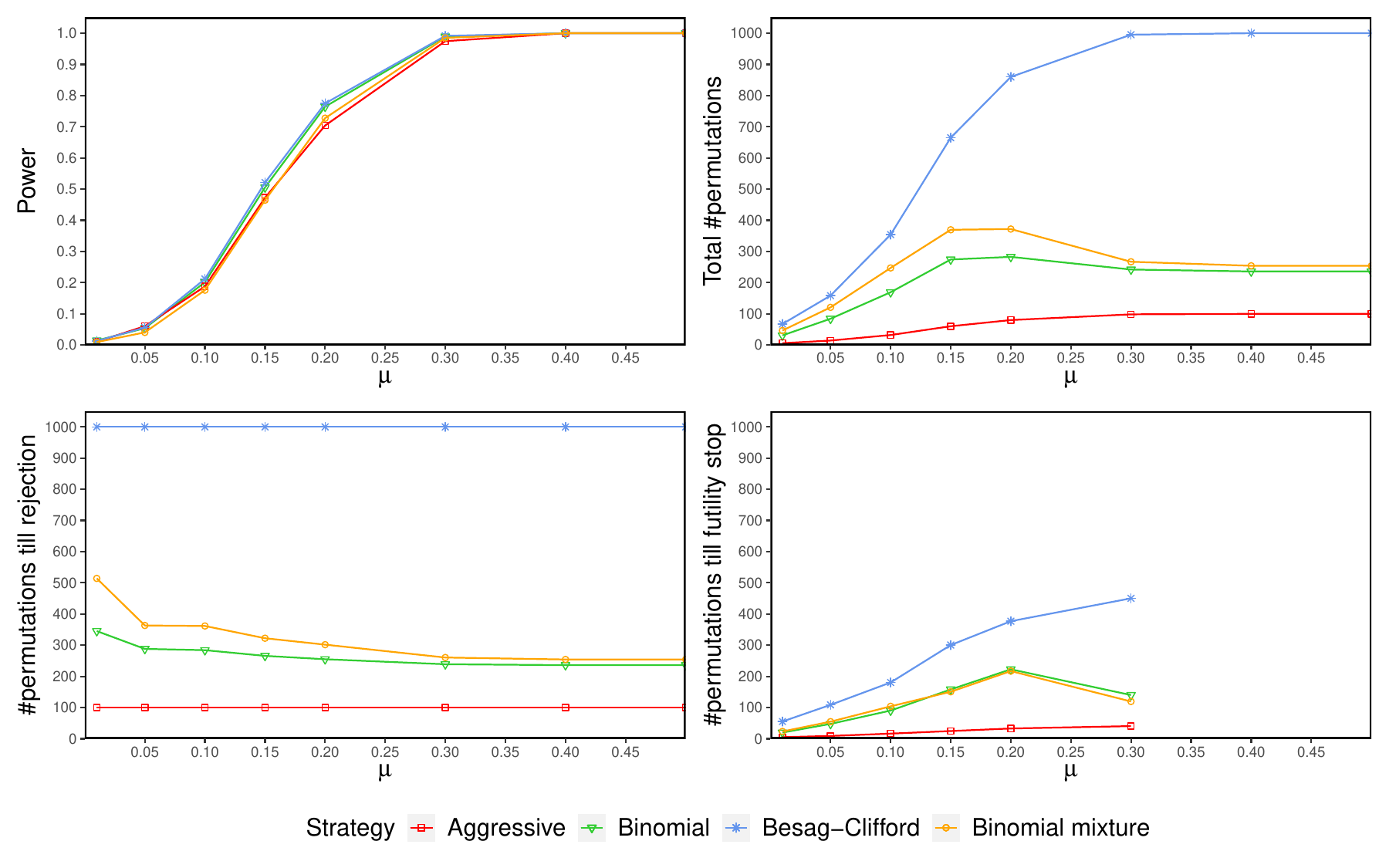}%
\caption{Identical caption as Figure~\ref{fig:sim_results_power}, except with $\alpha=0.01$. \label{fig:sim_results_power_2} }\end{figure}

\subsection{Improving power by stochastic rounding\label{sec:sim_stochastic_rounding}}

In this section, we quantify the gain in power obtained by stochastic rounding described in Section \ref{sec:stochastic_rounding}. Note that applying stochastic rounding to the aggressive strategy does not lead to an improvement, since it either rejects or has zero wealth when stopped, which is why we do not include it in this section. The simulation setup is the same as in Section \ref{sec:sim_power} and therefore the stopping times are the same as in Figures \ref{fig:sim_results_power} and \ref{fig:sim_results_power_2}. 

 Figure \ref{fig:sim_results_randomized_improvement} shows that stochastic rounding leads to a significant improvement in power, particularly for the binomial mixture strategy in case of $\alpha=0.01$. This also implies that the binomial mixture strategy was often stopped while having a promising wealth. Therefore, drawing $1000$ permutations might be not enough when the $\alpha$ is low and the decision is difficult ($\plim$ close to $\alpha$). Hence, the binomial mixture strategy is likely to achieve further rejections if we would continue sampling. In Figure \ref{fig:sim_results_randomized_improvement_2} we compare the power of these randomized strategies with the Besag-Clifford method. There is hardly any power difference visible. In particular, the Besag-Clifford method and randomized binomial mixture strategy seem to overlap completely. Note that our methods obtained this power while reducing the number of permutations considerably (Figure \ref{fig:sim_results_power} and Figure \ref{fig:sim_results_power_2}) and offering the option to stop at any point in time. 

 In Appendix~\ref{appn:additional_sim_results}, we provide additional simulations results comparing the randomized binomial mixture strategy with the Besag-Clifford method in terms of power and number of permutations for different parameter configurations. All results show the same behavior, the two methods provide a similar power, while the binomial mixture strategy reduces the number of permutations substantially.


\begin{figure}
\centering
\includegraphics[width=15cm]{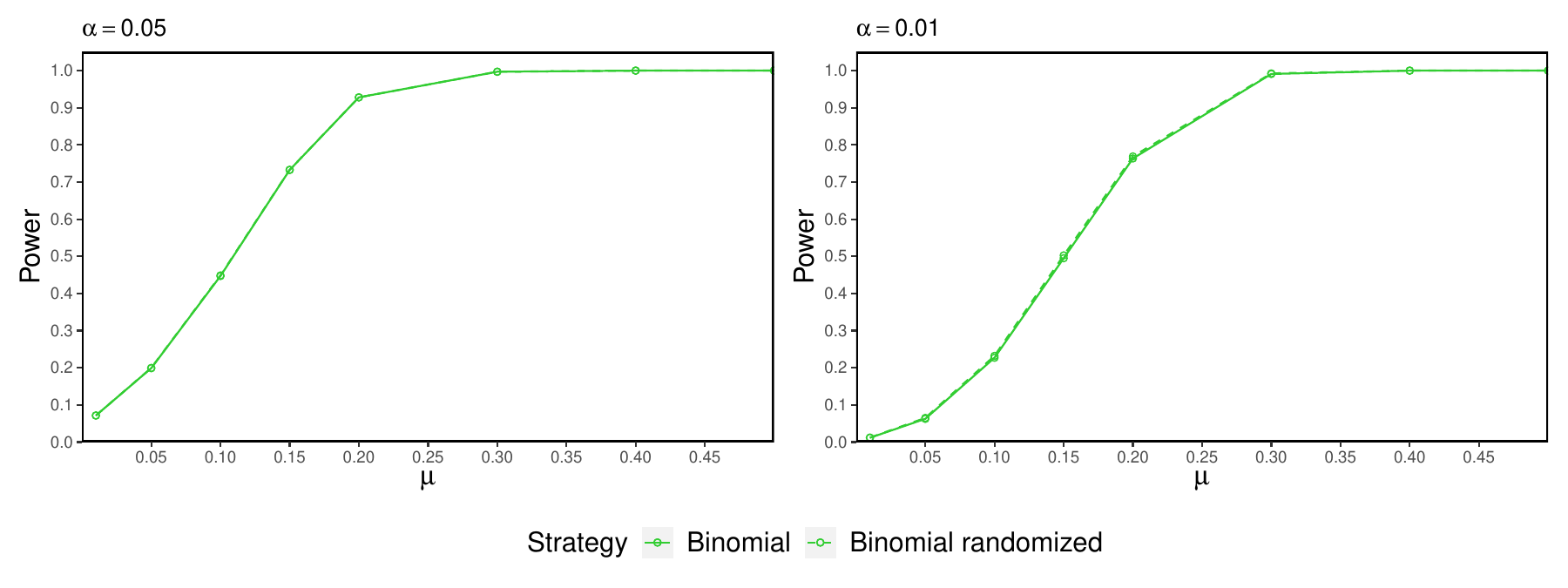}\\
\includegraphics[width=15cm]{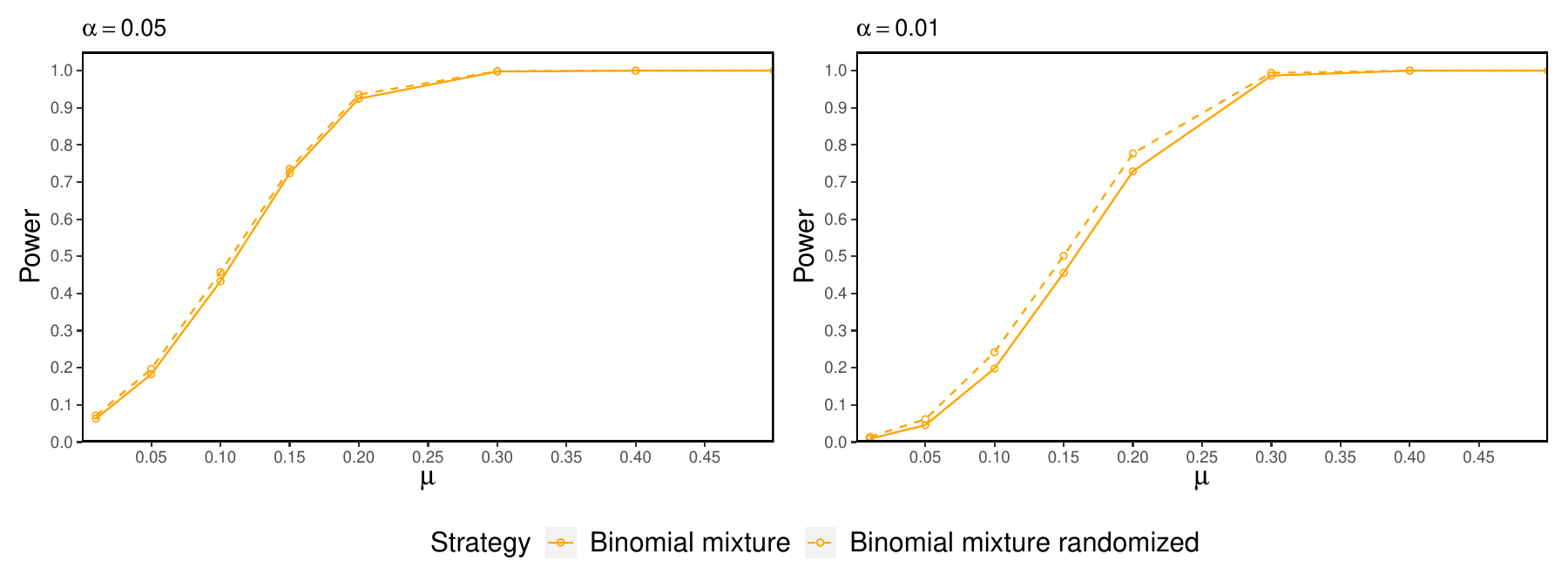}%
\caption{Power comparison of randomized and non-randomized strategies for different significance levels $\alpha$ and strengths of the alternative $\mu$ --- following the experimental protocol in Subsection \ref{sec:sim_stochastic_rounding}. The most significant power gain by the randomized strategies is obtained in the cases where the power loss compared to Besag-Clifford was largest (compare to Figure \ref{fig:sim_results_power} and \ref{fig:sim_results_power_2}).
\label{fig:sim_results_randomized_improvement} }
\end{figure}

\begin{figure}
\centering

\includegraphics[width=15cm]{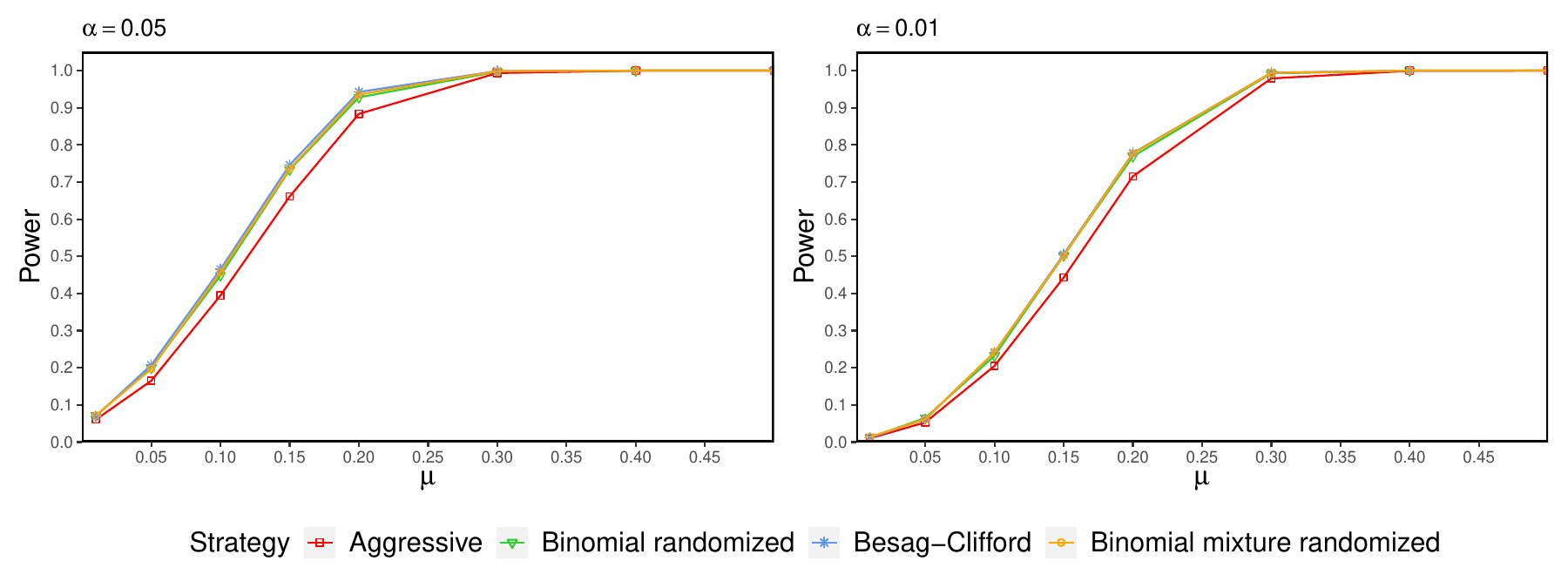}%
\caption{Power of randomized strategies and Besag-Clifford for different significance levels $\alpha$ and strengths of the alternative $\mu$ --- following the experimental protocol in Subsection \ref{sec:sim_stochastic_rounding}. The power of the randomized improvements of the binomial and binomial mixture strategy is very similar to the power obtained by Besag-Clifford, while the number of permutations is reduced significantly (see Figure \ref{fig:sim_results_power} and \ref{fig:sim_results_power_2}).
\label{fig:sim_results_randomized_improvement_2} }
\end{figure}

\subsection{Real data: testing Fisher's sharp null\label{sec:real_data}}

We consider a treatment vs.\ control trial with binary outcomes that was already used by \citet{ding2017paradox} and \citet{rosenbaum2002overt}. There were 18 successes among the 32 treated observations and 5 successes among the 21 control observations. We are interested in testing Fisher's sharp null hypothesis of no individual treatment effect $H_0^{\mathrm{IT}}$ at level $\alpha=0.05$. Under the Stable Unit Treatment Value Assumption (SUTVA; \cite{rubin1980randomization}), which states that there is only one version of the treatment and no interference between the individuals, this null hypothesis can be tested using the permutation test described in Example \ref{example:Fisher}.

\begin{figure}[h!]
\centering
\includegraphics[width=12cm]{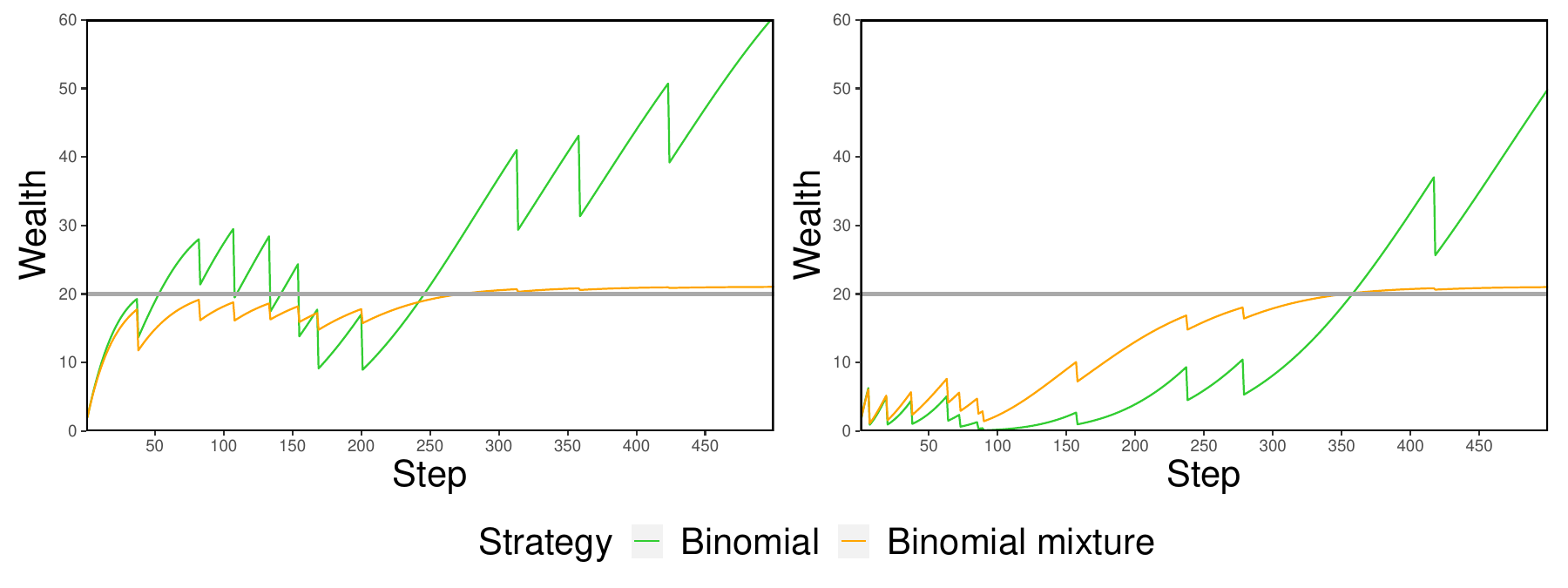}%
\\
\includegraphics[width=12cm]{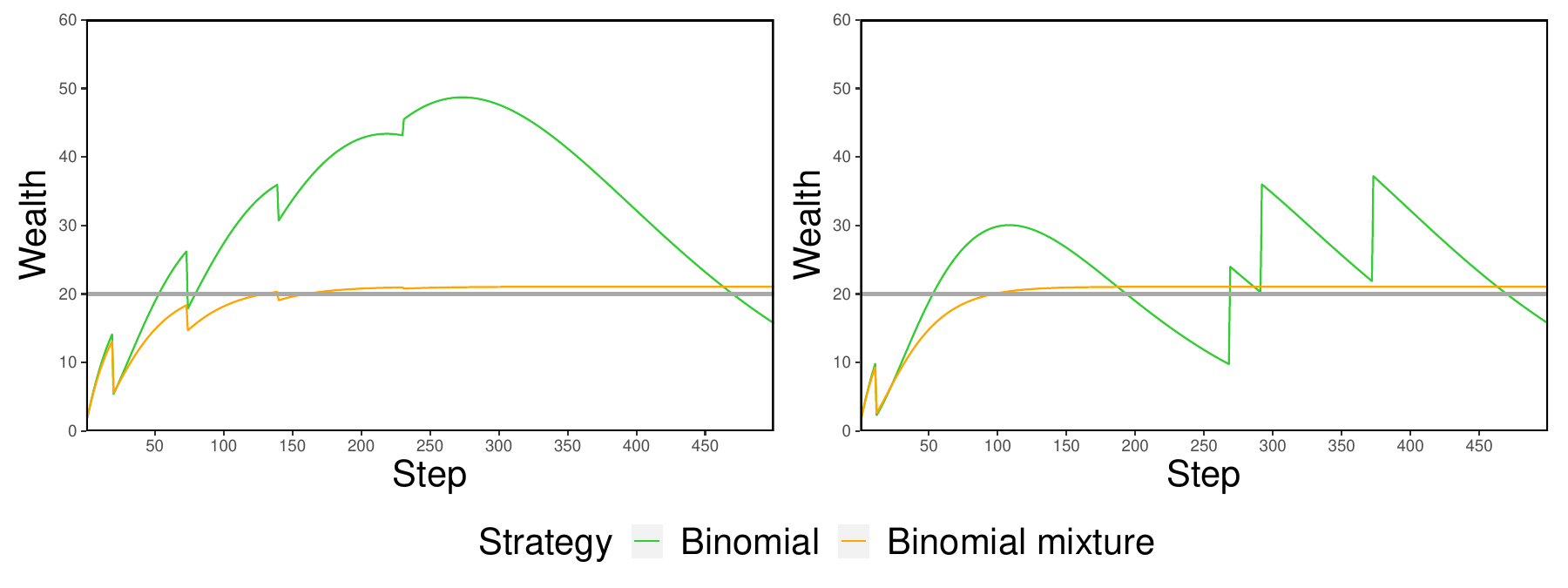}%
\caption{Wealth processes of the binomial strategy and binomial mixture strategy when applied on a real treatment vs.\ control trial for different runs --- following the experimental protocol in Subsection~\ref{sec:real_data}. The grey horizontal line equals $1/\alpha$.  At the beginning the increase in wealth is similar for both strategies. As the wealth of the binomial mixture strategy approaches the level $1/\alpha$, it begins to bet less aggressively (as expected).\label{fig:real_data_wealth_process} }\end{figure}

Fisher's exact test offers a p-value of $0.026$ and thus the null hypothesis could be rejected by this test. We applied our binomial strategy (Algorithm~\ref{alg:1}) and binomial mixture strategy with density $u_{c}$, $c=0.95\alpha$, $1000$ times  with a maximum number of $5000$ permutations. In this case, we omitted the stop for futility to not affect the power. The null hypothesis was rejected by both strategies in all $1000$ runs. The binomial strategy needed a mean number of $85$ permutations and a median number of $53$ permutations, while the binomial mixture strategy required a mean number of $147$ permutations and median number of $97$ permutations to obtain the decision. As a comparison, we also applied the Besag-Clifford method with $h=10$ and $T=h/\alpha$, leading to a mean (and median) number of $200$ permutations. Note that in this case the mean (and median) number of permutations is greater than with our sequential approaches. Still, the Besag-Clifford method failed to reject in $3$ of the $1000$ cases. If $h$ is reduced to $8$ (leading to a mean and median number of $160$ permutations), it even failed to reject in $16$ cases. The reason for this gain in efficiency by our methods is that they adapt the number of permutations to the difficulty of the decision, as also described in the previous subsection. Also note that the binomial mixture strategy even guarantees to reject at some finite time in this case, as the true p-value is below $c=0.0475$. 

We illustrated the wealth processes of our sequential strategies for different runs and the first $500$ steps in Figure~\ref{fig:real_data_wealth_process}. Since the binomial mixture strategy converges to $1/c$, we expect it to bet more safely when it approaches $1/\alpha$. For this reason, the binomial strategy allows to reject the hypothesis earlier in most cases. However, when there are a lot of losses, e.g. as in the beginning of the upper right plot, the binomial mixture strategy leads to a larger wealth and thus seeming more promising when the evidence is not very strong. Also note that when the number of wins is much larger than the number of losses, as in the lower plots, the binomial strategy begins to lose wealth when further wins occur --- after the hypothesis is already rejected. Due to \eqref{eq:log_opt_strategy}, this is the case if $(L_{t-1}+1)/(t+1)<\pchoice$, where $\pchoice=1/\lceil\sqrt{2\pi e^{1/6}}/\alpha\rceil$. 

It might seem redundant to reduce the number of permutations in this simple example, since it is not computationally expensive anyway. However, computational cost can be a serious issue in causal inference, particularly when rerandomization is needed \citep{morgan2012rerandomization}. To balance previously observed covariates between treatment groups, patients are often randomized multiple times until some prespecified balance criterion is met. In order to obtain an unbiased test for $H_0^{\mathrm{IT}}$, the same criterion need to apply for each of the permuted datasets. 
Thus, each permutation is checked for its acceptability and unacceptable permutations are discarded. This intermediate check can be very costly and lead to an enormous increase of the required number of permutations, especially when the number of covariates is large and/or the balance criterion is strict. Since the covariate balance is better with stricter criteria, reducing the number of permutations helps with balancing covariates between treatment groups \citep{morgan2015rerandomization}.

\subsection{Real data: testing conditional independence under model-X\label{sec:real_CRT}}

Here, we consider testing the independence between two variables $W_0$ and $X$ conditional on some covariates $Z$, as described in Example \ref{example:CRT}. We follow the analysis by \citet{grunwald2023anytime} and \citet{berrett2020conditional} and apply our methods on the Capital bikeshare dataset, which is available at \url{https://ride.capitalbikeshare.com/system-data}. In this example, $W_0$ denotes the logarithm of the ride duration, $X$ the binary variable membership (casual users vs.\ long-term membership) and $Z$ is the three dimensional vector of starting point, destination and starting time. For the generation of the test statistics $Y_0,Y_1,\ldots$, we used the code of \citet{grunwald2023anytime}, which is available in the supplementary material of their paper. The distribution $W_0|Z$ is modeled as Gaussian $N(\mu_Z,\sigma_Z^2)$, where $\mu_Z$ and $\sigma_Z^2$ are estimated by a kernel regression of $W_0$ with the time of day as covariate, separately for each combination of starting point and destination. The test statistics are then given by $Y_t=|\mathrm{cor}(X,W_t-\mathbb{E}[W_0|Z])|$. 
 
 We applied our binomial strategy and binomial mixture strategy $1000$ times with the same parameters as in Section~\ref{sec:real_data} (in particular, the type I error level was set to $0.05$). In the original code of \citet{grunwald2023anytime}, the training dataset consists of $158,741$ observations and the test data of $7,173$ observations. In this case, all $Y_i$, $i\geq 1$, were smaller than $Y_0$ such that the binomial strategy needed $44$ permutations and the binomial mixture strategy $61$ permutations to reject the null hypothesis. We also applied our methods when only the first $2000$ observations of the test data were included, thus increasing the probability of $Y_i$ being greater or equal than $Y_0$. This led to a mean permutation p-value of $0.008$ and the null hypothesis was still rejected by both sequential strategies in all $1000$ runs. The binomial strategy needed a mean number of $49$ permutations and a median number of $44$ permutations, while the binomial mixture strategy required a mean number of $83$ permutations and median number of $61$ permutations to obtain the decision. Since the evidence against $H_0$ is very large in this data, the classical permutation p-value and the Besag-Clifford method most likely also reject with a small number of maximum permutations $T$. However, note that we usually do not know how large the evidence in the data really is and therefore it would be a risky strategy to prespecify a small $T$. In contrast, our methods provide strong guarantees regarding the resampling risk while rejecting fast due to large evidence against the null.

\section{Discussion}
We have introduced a new framework for the construction of anytime-valid Monte-Carlo tests based on a simple testing by betting approach. Our general method encompasses all permutation tests based on the number of losses and particularly provides anytime-valid generalizations of the classical permutation p-value and the Besag-Clifford method. We derived the (oracle) log-optimal betting strategy under the considered alternative, given the true distribution of the limiting permutation p-value. We showed that for any given working prior on the limiting permutation p-value $\plim$, the mimicked log-optimal strategy can be written as a mixture of simple binomial strategies. Based on this, we proposed a strategy with zero resampling risk after a finite number of permutations for all $\ptrue \in [0,c)$, where $c<\alpha$ can be chosen arbitrarily close to $\alpha$. Furthermore, we performed a simulation study and a real data analysis which illustrates that our anytime-valid p-values have a similar power as the classical permutation p-value and the Besag-Clifford method while the number of permutations until the decision is obtained can be substantially reduced.

In Table \ref{tab:my_label} we list the pros and cons of the proposed betting strategies. To summarize, the aggressive strategy should be used when one is mainly concerned with reducing the number of permutations and a (small) loss of power can be accepted in return. However, in most cases the binomial and binomial mixture strategies are more appropriate, since they have a (slightly) larger power while still keeping the number of permutations quite low. If stochastic rounding is accepted, the binomial mixture strategy should be used, since it is the theoretically most grounded strategy and led to the best performance after the randomization. If stochastic rounding should be avoided, both the binomial and binomial mixture strategies are reasonable. The binomial strategy tends to stop earlier, while the binomial mixture strategy offers a better asymptotic behavior and can be used to bound the resampling risk. 

In future work, we hope to explore applications of our approach to multiple testing.

\begin{table}[h!]
    \centering
    \begin{tabular}{|l|l|l|}
    \hline
         & \textbf{Pros} & \textbf{Cons} \\ \hline
       \textbf{Aggressive strategy}  &
           \makecell[l]{\tabitem Requires least possible permutations \\
           \tabitem Parameter-free}
        & \makecell[l]{\tabitem Weak result regarding \\ 
        \hphantom{\tabitem} resampling risk \\ \tabitem Slightly less power than the \\ \hphantom{\tabitem} other strategies} 
        \\ \hline
       \textbf{Binomial strategy} &  \makecell[l]{\tabitem Reduces permutations substantially \\ \tabitem Similar power as classical permutation \\ \hphantom{\tabitem}  p-value in all simulation scenarios \\ \tabitem Always rejects earlier than classical \\
       \hphantom{\tabitem} permutation p-value, if \\ \hphantom{\tabitem} $\pperm < 1/\lceil \sqrt{2 \pi e^{1/6}}/\alpha\rceil \approx \alpha/e$ \\ 
       \hphantom{\tabitem} (Proposition \ref{prop:always_reject}) \\ \tabitem Parameter-free (see Algorithm \ref{alg:1})} & \makecell[l]{\tabitem Cannot ensure arbitrarily \\ \hphantom{\tabitem} small resampling risk}  \\ \hline
       \makecell[l]{\textbf{Binomial mixture} \\ \textbf{strategy}} & \makecell[l]{\tabitem Reduces permutations substantially   \\
       \tabitem Nearly identical power as classical \\ \hphantom{\tabitem}
       permutation p-value in simulations \\ \hphantom{\tabitem} after stochastic rounding\\
       \tabitem Can keep resampling risk arbitrarily \\ 
       \hphantom{\tabitem} small (see Corollary \ref{corol:rr_mixture} and \ref{corol:rr-random_mixture}) \\ 
       \makecell[l]{ \tabitem Parameter $c\in [0,1]$ has simple \\ \hphantom{\tabitem} interpretation with respect to\\\hphantom{\tabitem} 
        asymptotic behavior (Theorem \ref{theo:asymptotic_ap})}
       \\ 
       \makecell[l]{\tabitem Wealth is decreasing in number \\ \hphantom{\tabitem}  of losses} \\
       \tabitem Theoretically most grounded strategy} & \makecell[l]{\tabitem Slightly lower power   \\
       \hphantom{\tabitem} than the binomial strategy  \\
       \hphantom{\tabitem} in the  simulations without  \\
       \hphantom{\tabitem}  stochastic rounding} \\ \hline
    \end{tabular}
    \caption{Pros and cons of the proposed betting strategies.}
    \label{tab:my_label}
\end{table}

\subsection*{Acknowledgments}
AR thanks Leila Wehbe for posing the question of whether permutation tests can be stopped early under null and alternative, and for stimulating practical discussions. LF thanks Timothy Barry for practical feedback on the first version of the paper. LF and AR thank Peter Grünwald for helpful comments. LF and AR thank two anonymous reviewers for useful suggestions that significantly improved the paper. LF acknowledges funding by the Deutsche Forschungsgemeinschaft (DFG, German Research Foundation) – Project number 281474342/GRK2224/2. AR was funded by NSF grant DMS-2310718.

\bibliography{main}
\bibliographystyle{plainnat}

\begin{appendix}

\section{More details on related work}
\label{appsec:related-work}

For the particular problem of testing our null hypothesis that the observations are exchangeable, there exists no nontrivial test martingale for the canonical data filtration (generated by the observations) \citep{ramdas2022testing}. There are two different approaches to circumvent this issue. The first idea is to replace the data filtration by a coarser one. This means the information to be released at each step is restricted. This was used by Vovk and colleagues in a conformal prediction approach~\citep{vovk2003testing, vovk2005algorithmic,vovk2021testing}, and this is also the approach taken in the current paper (modulo some differences around randomization of ranks). The second approach is to construct an e-process which is not a test martingale, based on universal inference \citep{wasserman2020universal,ramdas2022testing}. The idea is to partition the composite null hypothesis, define a test martingale for each subset and take the infimum over these test martingales to yield an e-process. In terms of testing by betting, it could be interpreted as playing multiple games, where the gambler's wealth is the minimum wealth across all these games. In this case, the wealth process is adapted to the canonical data filtration, has expectation of at most one under the null at any stopping time in the original filtration of the data.

Very recently, other sequential tests of exchangeability were proposed, based on the former approach. \citet{saha2023testing} introduced a test martingale for Markovian alternatives that is based on \enquote{pairwise betting}, which means that two observations are observed and bet on simultaneously. 
Independently, \citet{koning2023online} introduced a sequential test for general group invariance (including exchangeability). In particular, he derives likelihood ratios for group invariance and defines a test martingale as running product of these. Similarly, \citet{lardy2024anytime} construct general tests of group invariance based on conformal prediction. Our test martingale is also based on coarsening the data filtration by not revealing the $Y_t$ itself, but the loss indicator $I_t$, and is thus somewhat related to this line of work. The main difference of our approach to the previously mentioned approaches is the alternative we test against.  \citet{vovk2021testing} focuses on change-point alternatives, \citet{ramdas2022testing} and \citet{saha2023testing} on Markovian alternatives, while \citet{koning2023online} considers general group invariance versus normal distributions with a location shift. In contrast, our paper focuses on the usual setting of permutation tests (where permutations rather than data are sampled sequentially), where under the alternative, only the original test statistic is special while all the permuted statistics are exchangeable by design. It should be noted that the above works give the mentioned applications as examples of their more general frameworks. The general frameworks may include (a modified version of) our approach as a special case, as described in Remark~\ref{remark:vovk} for Vovk's conformal prediction approach. However, none of these works make the connection to sequential permutation tests (where not the data but the permutations are drawn sequentially). We demonstrated in this paper that this special case is of great practical interest, simplified these general frameworks for our particular alternative to make the method accessible to a broader audience and constructed powerful betting strategies yielding test martingales with easy calculable closed forms.

\section{Calibrating $\pperm_T$ into an e-value\label{sec:calibrated_p}} 
An e-value $E$ is a random variable with expectation less than one under the null hypothesis. For example, our test martingale $W_{\tau}$ is an e-value for each stopping time $\tau$. A (non-random) function $h$ such that $h(E)$ is a p-value, where $E$ is an e-value, is called e-to-p calibrator.  We have already argued in Section~\ref{sec:aggressive_strat} that $h(E)=1/E$ is an e-to-p calibrator.  It can even even be shown that $h$ is the only admissible e-to-p calibrator, which means that there exists no other e-to-p calibrator $g$ with $g\geq h$ and $g\neq h$ \citep{vovk2021values}. Note that we restrict to non-random calibrators $h$. Otherwise, $h(E)=1/E$ could be improved by $\tilde{h}(E)=U/E$, where $U$ is standard uniformly distributed and independent of $E$. The fact that $\tilde{h}(E)$ defines a valid p-value follows from a randomized improvement of Markov's inequality \citep{ramdas2023randomized}.

Analogously, we can define p-to-e calibrators (in the following, we just call these calibrators). However, in this case there are many more possible choices of calibrators. Below, we show that classical calibrators are inadmissible in our setting and introduce a new class of admissible calibrators for such problems. 

 \citet{vovk2021values} showed that a decreasing function $h$ is an admissible calibrator if and only if $h$ is upper semicontinuous, $h(0)=0$ and $\int_0^1 h(u) d u=1$. However, this result is restricted to continuous p-values while the permutation p-value $\pperm_T$ only takes on discrete values. For example, a simple and often used calibrator class satisfying the above conditions is given by $h(p)=\kappa p^{\kappa-1}$ for some $\kappa \in (0,1)$. Note that, e.g., for $\kappa=1/2$ and $T=2$, we have
$$\EE[h(\pperm_T)]=\frac{1}{3} \frac{1}{2} \sqrt{3}+\frac{1}{3} \frac{1}{2} \sqrt{\frac{3}{2}}+\frac{1}{3} \frac{1}{2} \sqrt{1} =\frac{1}{6} \left(\sqrt{3}+\sqrt{\frac{3}{2}}+1\right) < 1,$$ which shows that $h$ is inadmissible in this case. In the following we provide a characterization of admissible calibrators for p-values with the same support as $\pperm_T$.

\begin{proposition}\label{prop:admissible_cal}
    Let $\pperm_T$ be a p-value with support $\{1/(T+1), \ldots,T/(T+1), 1\}$ and $\mathbb{P}(\pperm_T\leq r/(T+1))=r/(T+1)$ for all $1\leq r \leq T+1$. The calibrator $h$ is admissible for $\pperm_T$ if and only if $\sum_{r=1}^{T+1} h(r/(T+1))=T+1$. 
\end{proposition}
\begin{proof}
    The proposition follows immediately from $\EE[h(\pperm_T)]=\sum_{r=1}^{T+1}h(r/(T+1))/(T+1) $ under the null hypothesis.
\end{proof}

An example of an admissible calibrator for $\pperm_T$ is
\begin{align}
h(\pperm_T)=h(r/(T+1)) = (T+1) / (r v_{T+1}), \quad 1 \leq r \leq T+1,
\label{eq:calibrator_1}\end{align}
where $v_{T}=1+1/2+1/3+\dots 1/T$ is the $T$-th harmonic number. This is reminiscent of the Benjamini-Yekutieli
correction in multiple testing \citep{benjamini2001control}. Indeed, if these e-values (calibrated p-values) are plugged into the e-BH procedure \citep{wang2022false}, one obtains a method that is equivalent to applying the BY procedure on the original p-values.




Another option is to define
\begin{align}
h(\pperm_T)=h(r/(T+1)) = \frac{(T+1) / \sqrt{r}}{s_{T+1}} > \sqrt{T+1}/(2\sqrt{r}) = 1/(2\sqrt{\pperm_T}),
\label{eq:calibrator_2}\end{align}
where $s_{T+1} = \sum_{i=1}^{T+1} 1/\sqrt{i} < 2\sqrt{T+1}-1$, which shows that this is a uniform improvement of the calibrator by  \citet{vovk2021values} for $\kappa=1/2$.

\section{The Besag-Clifford p-value for $T=\infty$: a negative binomial permutation test\label{sec:neg_bin}}

The Besag-Clifford strategy also offers an exact p-value if we do not specify an upper bound $T$ of maximum permutations, meaning that the p-value $$\pval_{\gamma(h)}, \quad \text{where } \gamma(h):=\lim_{T \to \infty} \gamma(h,T), $$
is also valid and exact. 
While Besag and Clifford pointed that out, they, and also most follow up work, were mainly concerned with the case $T<\infty$ in order to obtain a closed (bounded) sampling scheme. Nevertheless, the maximum number of permutations drawn is still almost surely finite, and, conditional on the data, has finite expectation and variance. Indeed, conditional on $\plim=\ptrue$, $\gamma(r)$ is just a sum of (shifted) geometric random variables, or equivalently, $\gamma(r)-r$ is negative binomial with \enquote{success} probability $\ptrue$, so we call the test as the negative binomial permutation test and provide the result below for easier reference.

\begin{proposition}[Negative binomial permutation test]
For any pre-defined number of losses $h \geq 1$, define $\gamma(h) := \inf\{t \geq 1: \sum_{i=1}^t \mathbbm{1}\{Y_t \geq Y_0\} \geq h\}$ as being the time of the $h$-th loss. Then, $\pval_{\gamma(h)} := h/\gamma(h)$ is a valid and exact p-value for $H_0$: for any $N \geq 0$, $\mathbb{P}_{H_0}(\pval_{\gamma(h)}\leq h/(h+N)) = h/(h+N)$.
\end{proposition}
\begin{proof}
Note that  $\mathrm{supp}(\pval_{\gamma(h)})=\{1,h/(h+1),h/(h+2),\ldots\}$ and for every $N\in \mathbb{N}_0$, we have $\mathbb{P}_{H_0}(\pval_{\gamma(h)}\leq h/(h+N))=\mathbb{P}_{H_0}(\gamma(h) > h+N-1)=h/(h+N)$, where the last equation follows from the fact that $\gamma(h)>h+N-1$ means that $Y_0$ is among the $h$ largest test statistics of $Y_0,Y_1,\ldots,Y_{h+N-1}$. 
\end{proof}

We remark that $\pval_{\gamma(h)}$ bears an interesting resemblance to the standard permutation p-value $\pperm_T$. Both of them equal the number of losses divided by the number of permutations; $\pperm_T$ fixes the denominator to equal $T$ and lets the numerator be determined by the data, while  $\pval_{\gamma(h)}$ fixes the numerator as $h$, and lets the denominator be determined by the data. The fact most relevant for us is that since
\[
\pagg_{\tau} =  \pagg_{\gamma(1)},
\]
the same exactness of p-value holds for $\pagg_\tau$. Said differently, one can obtain a valid and exact p-value as the inverse of the first time at which a permuted statistic exceeds the original unpermuted statistic. We summarize this in another proposition for ease of future reference.

\begin{proposition}
    If the aggressive strategy is stopped when its wealth hits zero, as defined by the time $\gamma(1)$, the resulting stopped p-value equals $1/\gamma(1)$, and is an exact p-value.  Further, if one stops at some arbitrary stopping time $\rho \leq \gamma(1)$, then  $1/\rho$ is also a valid (potentially inexact) p-value.
\end{proposition}

Thus, if a level $\alpha$ is prespecified, one can stop as soon as $1/\rho$ drops below $\alpha$ (or the first loss occurs, whichever happens first).
The above propositions slightly generalize the results by~\citet{besag1991sequential} about $\pbc_{\gamma(1,T)}$, because it can potentially stop at any arbitrary stopping time before $\gamma_{1,T}$. 


\section{Nonstandard asymptotic behavior of log-optimal bets and wealth\label{appn:asymp_log_opt}}

In the following, we characterize the long-term behavior of the log-optimal strategy and the corresponding logarithmic wealth. 

\begin{theorem}
    Let $B_t^*$ be the log-optimal betting strategy \eqref{eq:log_opt_strategy} at step $t$. It holds that $B_t^*(I_t) \to 1$ almost surely for $t\to \infty$.
    \label{theo:constant_bet}
\end{theorem}
This result means that even the log-optimal betting strategy will eventually have its wealth multiplied a factor arbitrarily close to one. This is a key aspect where our problem setting differs from many other standard i.i.d.-like problems --- for the latter, the optimal strategies usually multiply the wealth of the gambler by a constant factor strictly larger than one in expectation. The reason for this unusual behavior is that in our case the indicators $I_1, I_2,\ldots$ are exchangeable under both, the null and the alternative. Therefore, the null and alternative hypothesis learn the same distribution of $I_t$ in the long run. Thus, even if $H_1$ was true, at some point we can longer exploit this knowledge by betting against $H_0$.

\begin{proof}
      From Section~\ref{sec:resampling_risk}, we know that $I_1,I_2,\ldots$ are i.i.d. conditional on $\plim$ with $\mathbb{E}(I_1|\plim)=\plim$. Furthermore,  for $\pt$ from Proposition~\ref{theo:log_optimal}, we have $\pt=\mathbb{E}(I_t|I_1^{t-1})=\mathbb{E}(I_1|I_2^{t})$, since $I_1,\ldots,I_t$ are exchangeable. Levy's zero-one law implies that $\mathbb{E}(I_1|I_2^{t}) \stackrel{a.s.}{\to }\mathbb{E}(I_1|I_2^{\infty})=\mathbb{E}(I_1|\plim,I_2^{\infty})=\mathbb{E}(I_1|\plim)=\plim$ for $t \to \infty$ and thus $\pt \stackrel{a.s.}{\to } \plim$. We now show that $B_t^*(I_t)|\{\plim=\ptrue\} \stackrel{a.s.}{\to } 1$  for every possible realization $\ptrue$ of $\plim$. First, consider the case $\ptrue\in (0,1)$. Then $\pt (t+1)/(L_{t-1}+1)|\{\plim=\ptrue\} \stackrel{a.s.}{\to } \ptrue \cdot 1/\ptrue=1$ and $(1-\pt)(t+1)/(t-L_{t-1})|\{\plim=\ptrue\} \stackrel{a.s.}{\to } (1-\ptrue) \cdot 1/(1-\ptrue)=1$  for $t \to \infty$ and due to Proposition~\ref{theo:log_optimal}, we have $B_t^*(I_t)|\{\plim=\ptrue\} \stackrel{a.s.}{\to } 1$. Now consider the case $\ptrue=0$. Since $\mathbb{P}(I_t=1|\plim=0)=0$ for all $t\in \mathbb{N}$, we have $L_t|\{\plim=0\} = 0$ for all $t\in \mathbb{N}$ almost surely. Since  $(1-\pt)(t+1)/(t-L_{t-1}) |\{\plim=0\} \stackrel{a.s.}{\to } 1$, we have $B_t^*(I_t)|\{\plim=0\}\stackrel{a.s.}{\to } 1$. If $\ptrue=1$, we analogously have $L_t |\{\plim=1\} =t$ almost surely and $\pt (t+1)/(L_{t-1}+1)|\{\plim=1\} \stackrel{a.s.}{\to } 1$.
\end{proof}

\begin{corol}
    Let $W_T^*$ be the wealth obtained by the log-optimal betting strategy after $T$ permutations. It holds $\log(W_T^*)/T \to 0$ for $T\to \infty$ almost surely.
    \label{corol:log_opt_wealth}
\end{corol}

Again, for other i.i.d.-like problems in this area, the wealth grows exponentially and $\log(W_T^*)/T$ usually tends to a constant larger than 0 under the alternative. In our problem however, the wealth (and its maximum) will hit a plateau, and at some point sampling more permutations does not help anymore (as would be expected).

\begin{proof}
   Let $\epsilon>0$. Due to Theorem~\ref{theo:constant_bet}, there exists $N\in \mathbb{N}$ such that $|\log(B_t^*(I_t))|<\epsilon$ a.s. for all $t\geq N$. Therefore, 
   $$|\log(W_T^*)|\leq \sum_{t=1}^{N-1}|\log(B_t^*(I_t))| + \sum_{t=N}^T |\log(B_t^*(I_t))| < a + \epsilon T$$
   a.s. for all $T>N$, where $a$ is some non-negative constant. Hence, there exists $n\in \mathbb{N}$ such that $|\log(W_T^*)/T|< \epsilon$ a.s. for all $T\geq n$. 
\end{proof}

Theorem~\ref{theo:constant_bet} and Corollary~\ref{corol:log_opt_wealth} show that at a late stage we can only gain wealth very slowly. Therefore, we need to bet using the structure of the alternative carefully in order to be able to build up a significant wealth in an appropriate time period. In the following section, we introduce a simple strategy based on these results.

\section{Using a Beta distribution as working prior for $\plim$\label{sec:beta_dist}}
A common prior for the probability in a Bernoulli trial is a beta distribution. The probability density function of a beta distribution for shape parameters $a,b>0$ is given by 
$$ g_{a,b}(x)= \frac{x^{a-1} (1-x)^{b-1}}{\mathrm{Beta}(a,b)}, \text{ where } \mathrm{Beta}(a,b)=\int_0^1u^{a-1} (1-u)^{b-1} du \quad (0\leq x \leq 1).$$
The function $\mathrm{Beta}(a,b)$ is called beta function.
In this section we derive the wealth obtained by applying the binomial mixture strategy with a beta distribution as working prior for $\plim$ and show some simulation results for this strategy. 

\begin{proposition}
    After $T$ permutations and $\ell$ losses, the wealth obtained by the binomial mixture strategy with density $g_{a,b}$ is given by $W_T^{g_{a,b}}(\ell)=\mathrm{Beta}(a+\ell,b+T-\ell)/(\mathrm{Beta}(\ell+1,T-\ell+1)\mathrm{Beta}(a,b))$.
\end{proposition}
\begin{proof}
    First note that $(T+1){T \choose \ell}=\mathrm{Beta}(\ell+1,T+1)$. Therefore, the wealth is given by 
    \begin{align*}
        W_T^{g_{a,b}}(\ell)&=(T+1){T \choose \ell} \int_0^1 p^\ell (1-p)^{T-\ell} g_{a,b}(p) dp \\
        &= (T+1){T \choose \ell} \frac{\int_0^1 p^{\ell+a-1} (1-p)^{T+b-\ell-1} dp}{\mathrm{Beta}(a,b)} \\
        &=\frac{\mathrm{Beta}(\ell+a,T+b-\ell)}{\mathrm{Beta}(\ell+1,T-\ell+1)\mathrm{Beta}(a,b)}.
    \end{align*}
\end{proof}
It can be shown that if $a\leq 1$ and $b\geq 1$, then $g_{a,b}(x)$ is nonincreasing in $x$ which seems reasonable as working prior for $\plim$. In particular, for $a=1$ and $b\in \mathbb{N}$, the wealth simplifies to 
$$W_T^{g_{1,b}}(\ell)=b\frac{(T+1)\cdots(T-\ell+1)}{(T+b)\cdots(T-\ell+b)},$$
which is decreasing for an increasing number of losses. In Figure~\ref{fig:sim_results_beta} the log-transformed p-values obtained by the binomial mixture strategy with density $g_{1,b}$ are compared with the ones obtained by the binomial strategy, the binomial mixture strategy (with density $u_{0.9\alpha}$), the aggressive strategy and the classical permutation p-value. We used the same simulation setup as described in Section \ref{sec:simulations}. The binomial mixture strategy with beta working prior bets more aggressively when $b$ is large. However, the binomial strategy and the binomial mixture strategy with uniform working prior lead to more rejections, as seen by the larger number of log-transformed p-values below the grey line $\log(\alpha)$, in all cases.

\begin{figure}[h!]
\centering
\includegraphics[width=15cm]{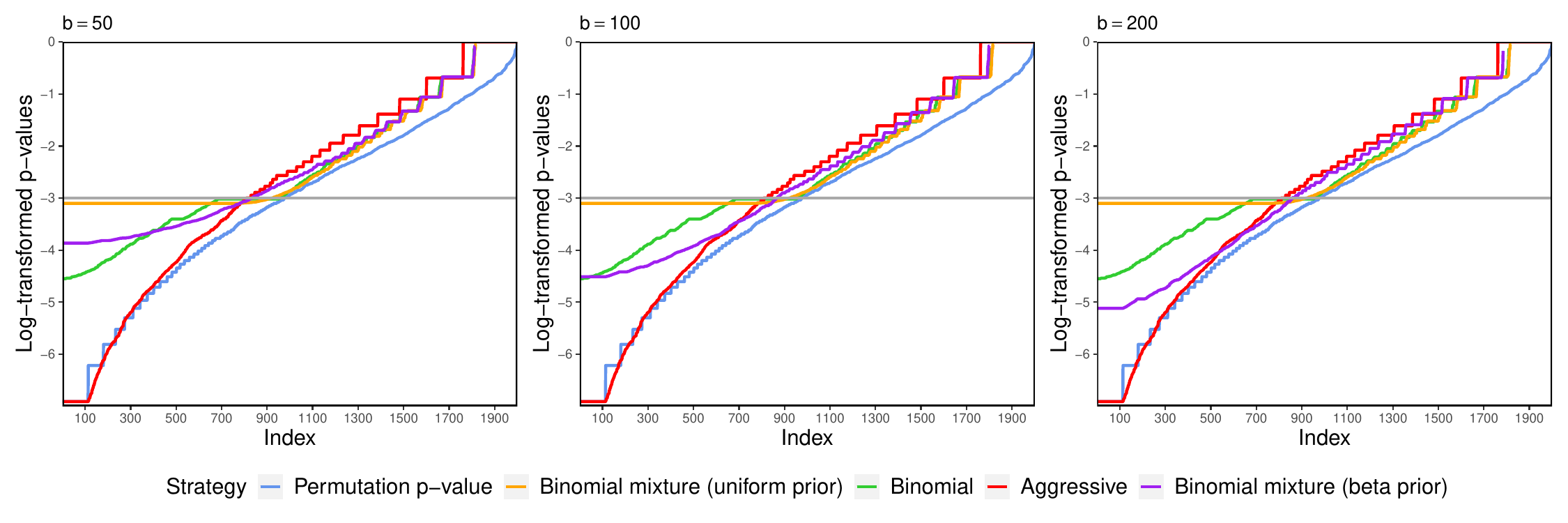}
\caption{All plots display log-transformed p-values in ascending order for $2000$ simulations --- following the experimental protocol in Subsection~\ref{sec:beta_dist}. The grey horizontal line equals $\log(\alpha)$, where $\alpha$ equals 0.05. The binomial mixture strategy with a beta working prior bets more aggressive when $b$ is large. However, in terms of rejections (p-values below grey line) it is outperformed by the binomial strategy and the binomial mixture strategy with uniform working prior in all cases. \label{fig:sim_results_beta} }\end{figure}

\section{Additional simulation results\label{appn:additional_sim_results}}

In this section, we provide additional simulation results for the comparison of the Besag-Clifford method with our randomized binomial mixture strategy in terms of power and number of permutations. The  simulation setup is the same as described in Section \ref{sec:simulations}, however, in each figure we varied one of the parameters. In Figure \ref{fig:sim_results_log_normal}, the observations follow a log-normal distribution (control observations were drawn from $\exp(Z)$ and treated observations from $\exp(Z+\mu)$, where $Z$ follows a standard normal distribution). In Figures \ref{fig:sim_results_c0.8} and \ref{fig:sim_results_c0.99}, we chose the parameter of the binomial mixture strategy as $c=0.8\alpha$ and $c=0.99\alpha$, respectively. In Figures \ref{fig:sim_results_n100}--\ref{fig:sim_results_n2000}, the sample size in each trial was varied ($100$, $200$, $500$ and $2000$).

All figures show the same behavior as in Section~\ref{sec:simulations}: the power of both methods is nearly identical, while the binomial mixture strategy reduces the number of permutations substantially.

\begin{figure}[h!]
\centering
\includegraphics[width=15cm]{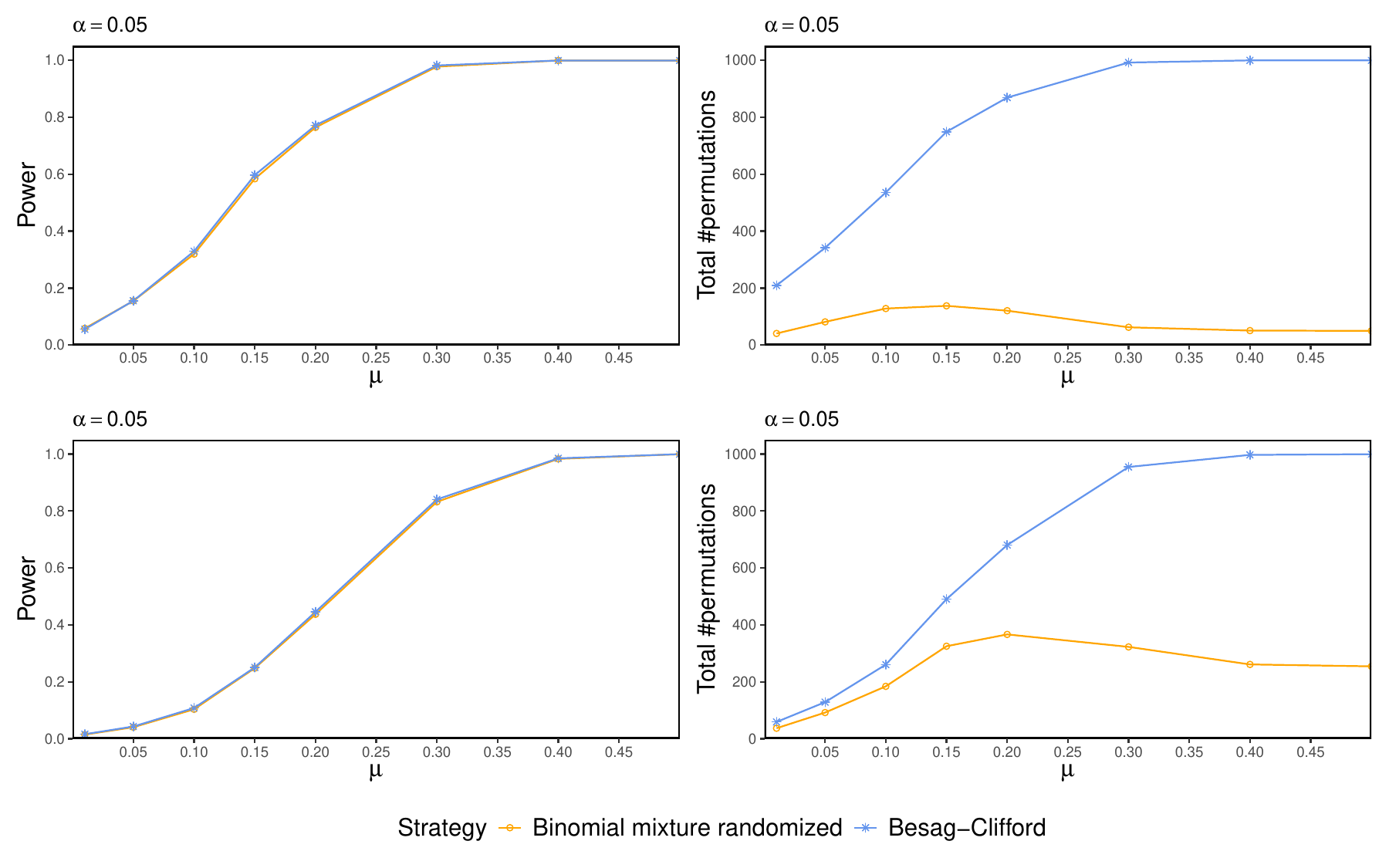}
\caption{Power and average number of permutations until the decision was obtained for log-normally distributed data --- following the experimental protocol in Appendix~\ref{appn:additional_sim_results}. \label{fig:sim_results_log_normal} }\end{figure}

\begin{figure}[h!]
\centering
\includegraphics[width=15cm]{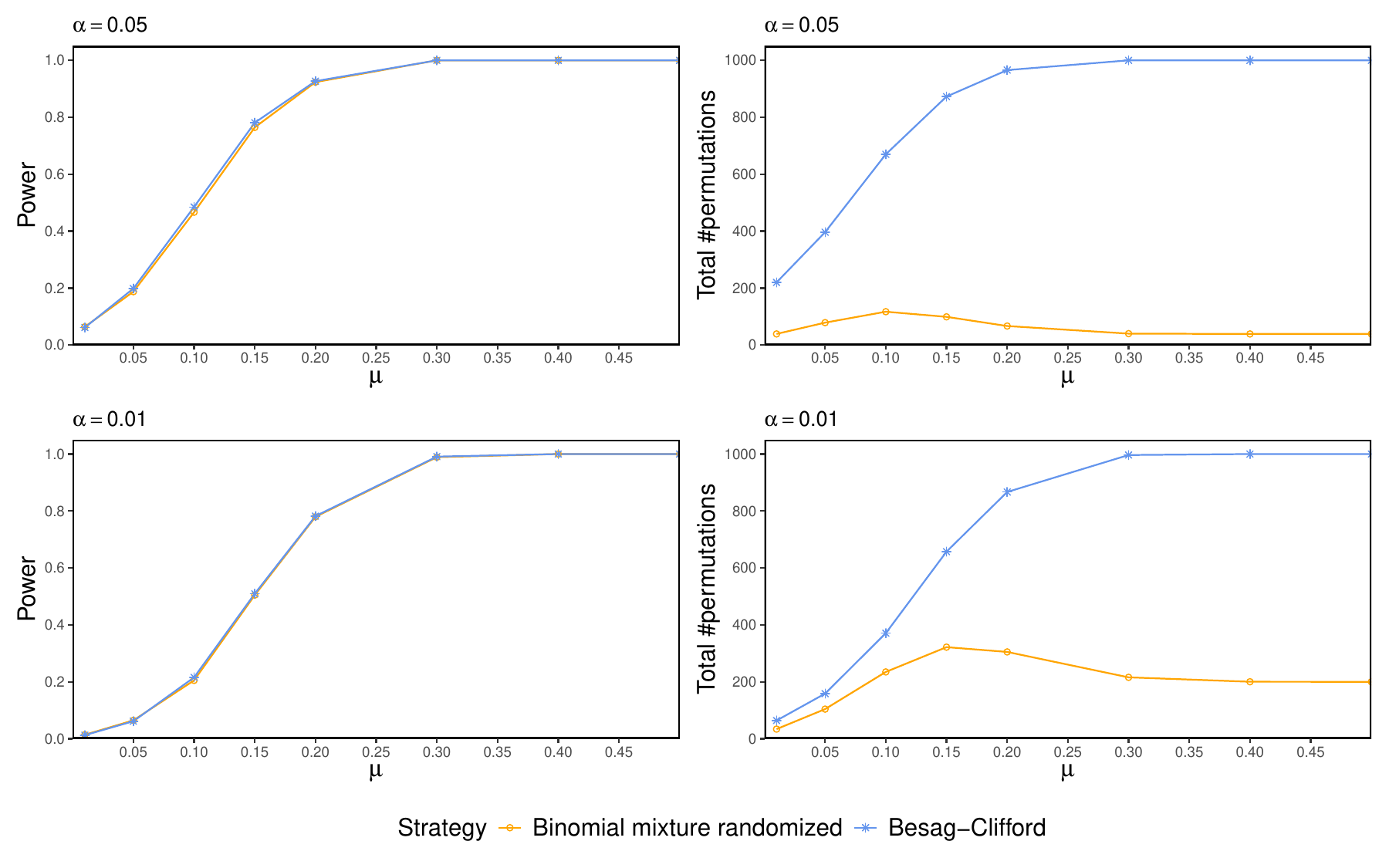}
\caption{Power and average number of permutations until the decision was obtained for $c=0.8$ --- following the experimental protocol in Appendix~\ref{appn:additional_sim_results}. \label{fig:sim_results_c0.8} }\end{figure}

\begin{figure}[h!]
\centering
\includegraphics[width=15cm]{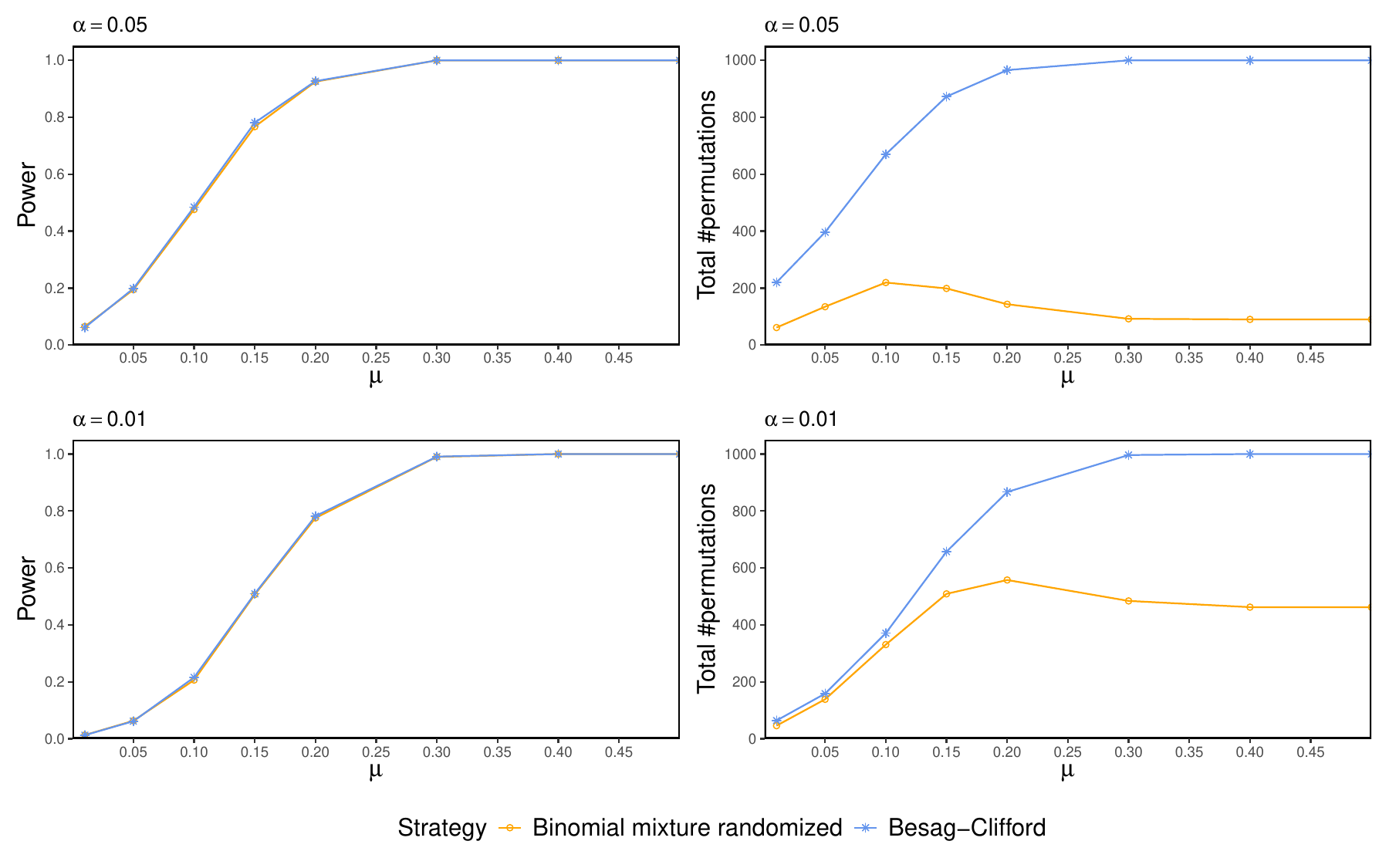}
\caption{Power and average number of permutations until the decision was obtained for $c=0.99$ --- following the experimental protocol in Appendix~\ref{appn:additional_sim_results}. \label{fig:sim_results_c0.99} }\end{figure}

\begin{figure}[h!]
\centering
\includegraphics[width=15cm]{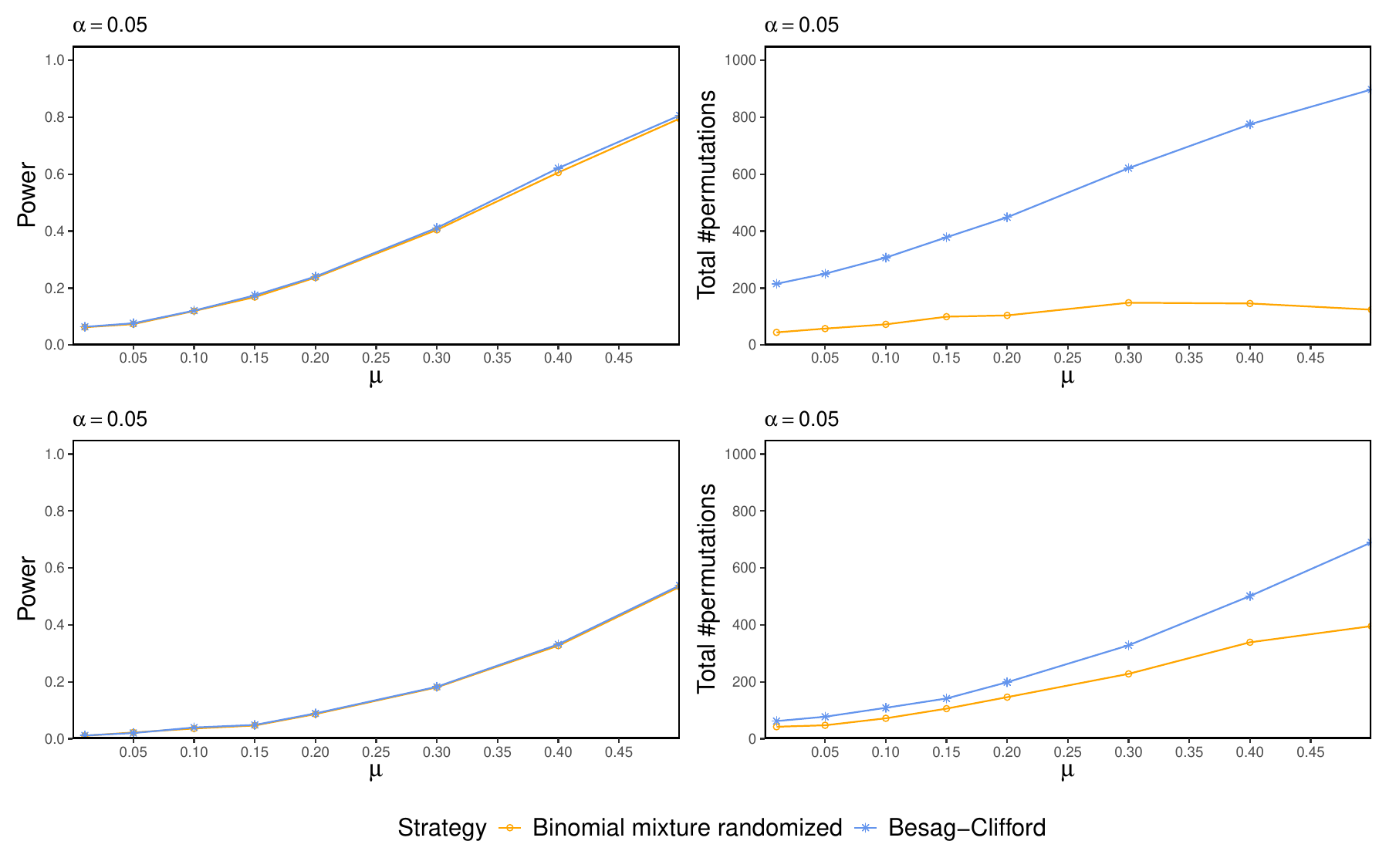}
\caption{Power and average number of permutations until the decision was obtained for $n=100$ --- following the experimental protocol in Appendix~\ref{appn:additional_sim_results}. \label{fig:sim_results_n100} }\end{figure}

\begin{figure}[h!]
\centering
\includegraphics[width=15cm]{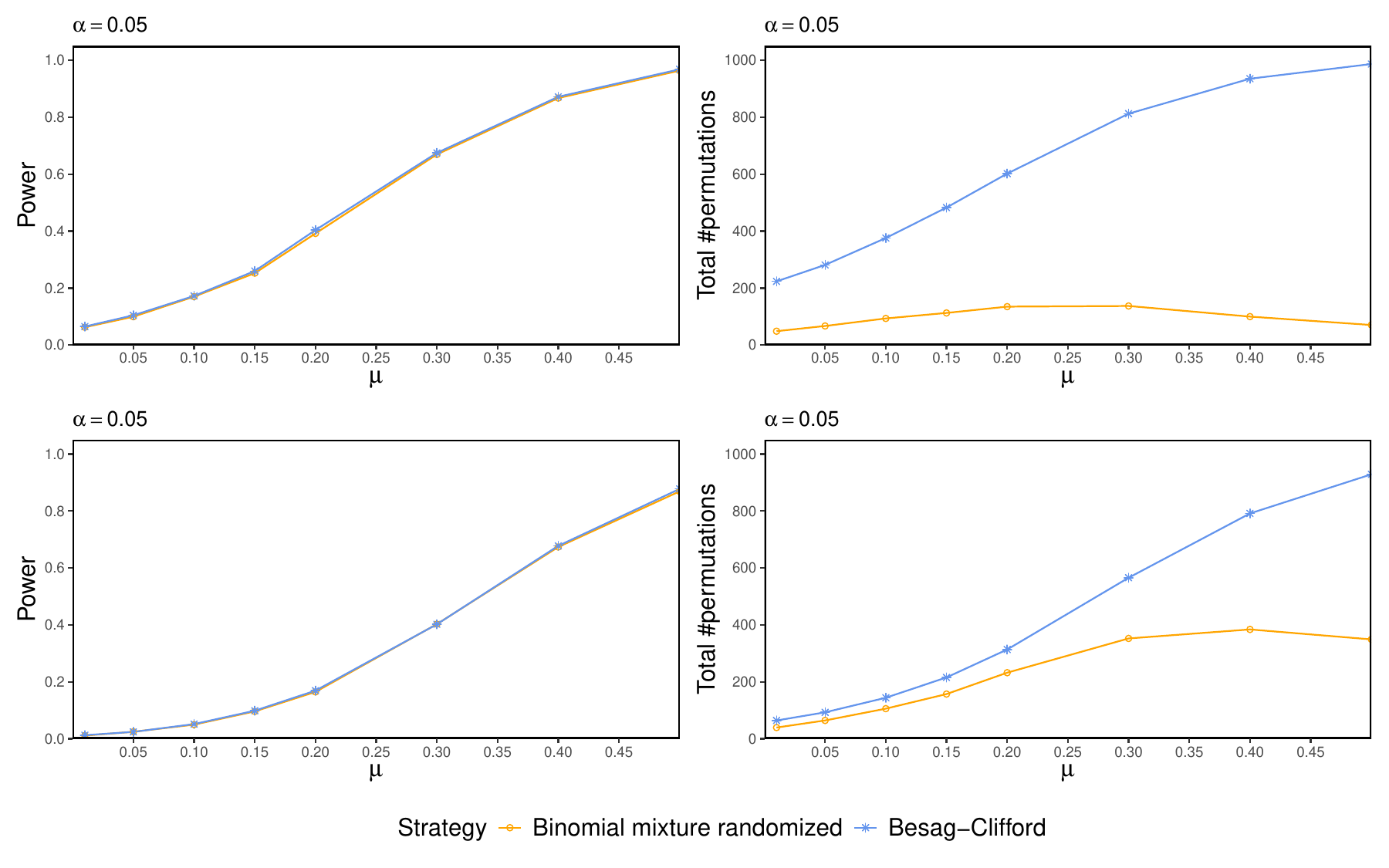}
\caption{Power and average number of permutations until the decision was obtained for $n=200$ --- following the experimental protocol in Appendix~\ref{appn:additional_sim_results}. \label{fig:sim_results_n200} }\end{figure}

\begin{figure}[h!]
\centering
\includegraphics[width=15cm]{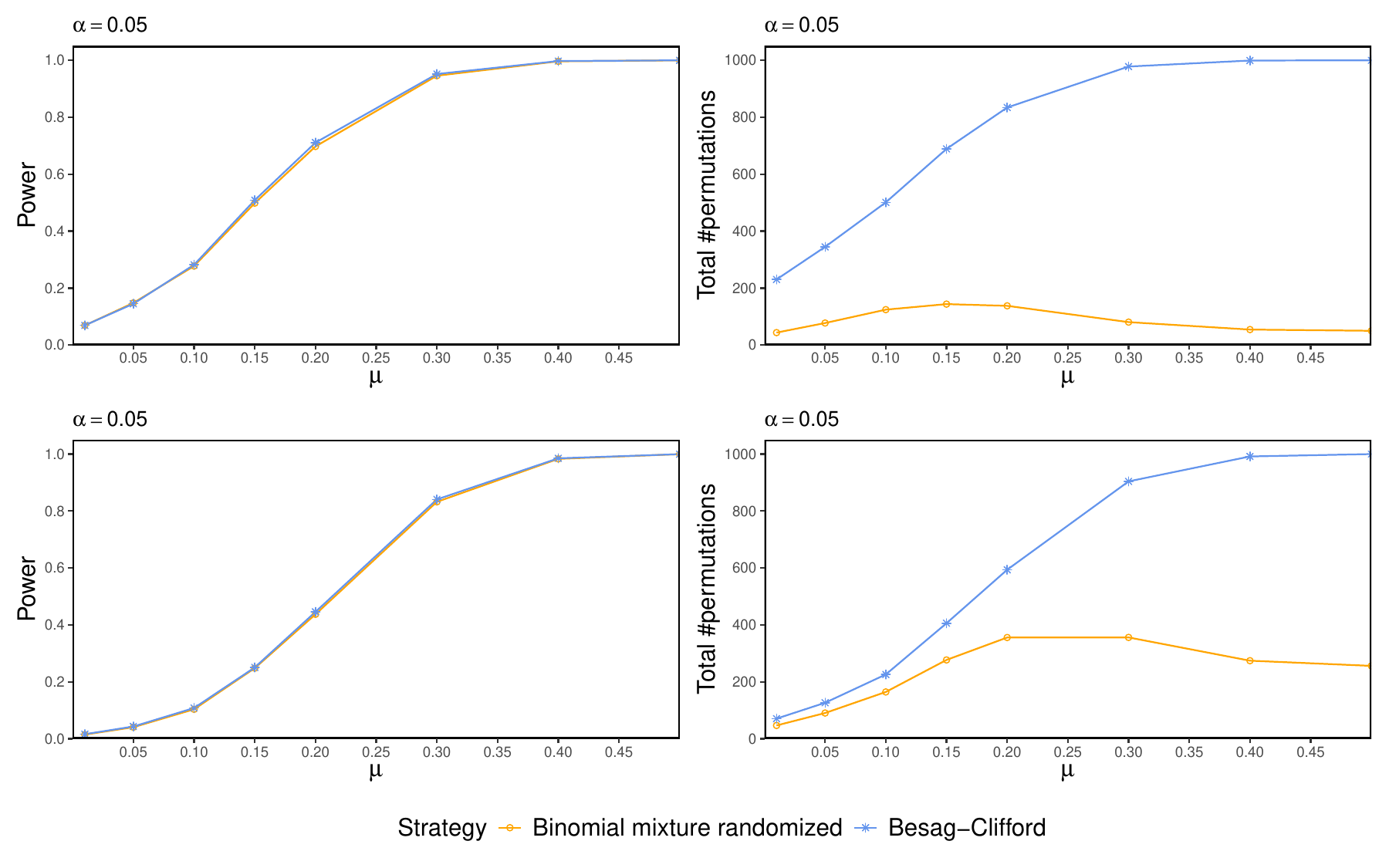}
\caption{Power and average number of permutations until the decision was obtained for $n=500$ --- following the experimental protocol in Appendix~\ref{appn:additional_sim_results}. \label{fig:sim_results_n500} }\end{figure}

\begin{figure}[h!]
\centering
\includegraphics[width=15cm]{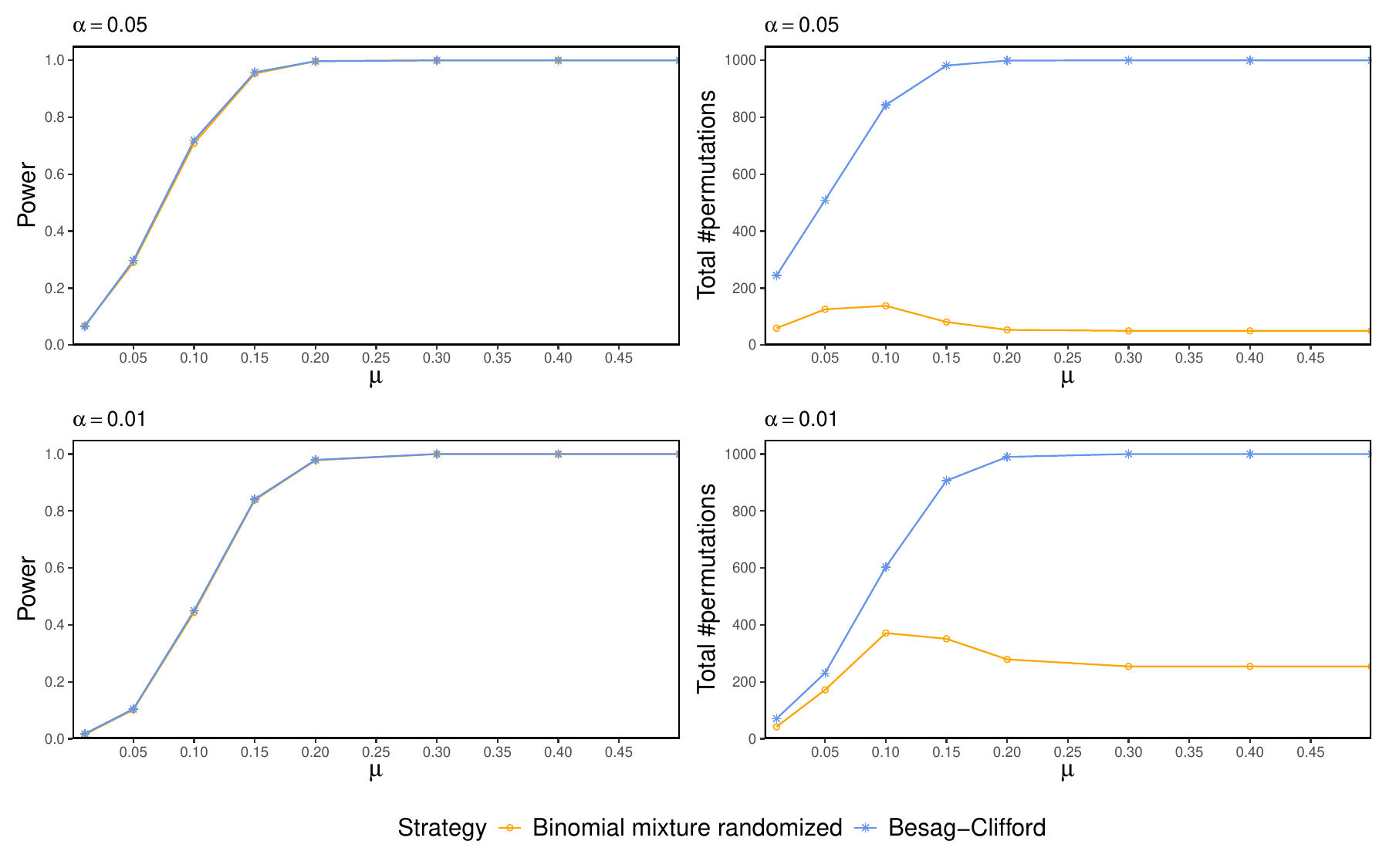}
\caption{Power and average number of permutations until the decision was obtained for $n=2000$ --- following the experimental protocol in Appendix~\ref{appn:additional_sim_results}. \label{fig:sim_results_n2000} }\end{figure}

\end{appendix}

\end{document}